\documentclass[11pt]{article}
\usepackage[margin=1in]{geometry}
\usepackage{graphicx} 
\usepackage{amsmath,amssymb,amsthm}
\usepackage{graphicx}
\usepackage[colorlinks=true, allcolors=blue]{hyperref}
\usepackage[plain,noend]{algorithm2e}
\usepackage{bbm}
\usepackage[authoryear,round]{natbib}
\makeatletter
\renewcommand{\algocf@captiontext}[2]{#1\algocf@typo. \AlCapFnt{}#2} 
\def\@algocf@capt@plain{top}
\renewcommand{\algocf@makecaption}[2]{%
  \addtolength{\hsize}{\algomargin}%
  \sbox\@tempboxa{\algocf@captiontext{#1}{#2}}%
  \ifdim\wd\@tempboxa >\hsize
    \hskip .5\algomargin%
    \parbox[t]{\hsize}{\algocf@captiontext{#1}{#2}}
  \else%
    \global\@minipagefalse%
    \hbox to\hsize{\box\@tempboxa}
  \fi%
  \addtolength{\hsize}{-\algomargin}%
}
\makeatother

\newcommand{\bV}{ {\bf V} }

\newcommand{\sX}{ { \mathbb{ X}} }
\newcommand{\sA}{ { \mathcal{ A}} }
\newcommand{\sP}{ { \mathcal{ P}} }

\newcommand{\kl}{\mathrm{KL}}

\DeclareMathOperator{\sign}{\mathrm{sign}}
\newcommand{\R}{\mathbb{R}}
\newcommand{\N}{\mathbb{N}}
\newcommand{\V}{\mathbb{V}}
\newcommand{\E}{\mathbb{E}}
\newcommand{\I}{\mathbb{I}}

\newtheorem{assumption}{Assumption}

\newtheorem{remark}{Remark}
\newtheorem{lemma}{Lemma}
\newtheorem{theorem}{Theorem}
\newtheorem{proposition}{Proposition}
\newtheorem{corollary}{Corollary}

\title{Zeroth-order parallel sampling}
\date{\vspace{-5ex}}
\author{Francesco Pozza\thanks{Bocconi Institute for Data Science and Analytics, Bocconi University, Milan, Italy. \texttt{francesco.pozza2@unibocconi.it}}\;  and Giacomo Zanella\thanks{
Department of Decision Sciences and Bocconi Institute for Data Science and Analytics, Bocconi University, Milan, Italy. \texttt{giacomo.zanella@unibocconi.it}\\
\emph{Both authors acknowledge support from the European Union (ERC), through the Starting grant `PrSc-HDBayLe', project number 101076564.}}
}

\begin{document}

\maketitle
\begin{abstract}
Finding effective ways to exploit parallel computing to accelerate Markov chain Monte Carlo methods is an important problem in Bayesian computation and related disciplines. In this paper, we consider the zeroth-order setting where the unnormalized target distribution can be evaluated but its gradient is unavailable for theoretical, practical, or computational reasons. We also assume access to $m$ parallel processors to accelerate convergence. The proposed approach is inspired by modern zeroth-order optimization methods, which mimic gradient-based schemes by replacing the gradient with a zeroth-order stochastic gradient estimator. Our contribution is twofold. 
First, we show that a naive application of popular zeroth-order stochastic gradient estimators within Markov chain Monte Carlo methods leads to algorithms with poor dependence on $m$, both for unadjusted and Metropolis-adjusted schemes. 
We then propose a simple remedy to this problem, based on a random-slice perspective, as opposed to a stochastic gradient one, obtaining a new class of zeroth-order samplers that provably achieve a polynomial speed-up in $m$. Theoretical findings are supported by numerical studies.
\end{abstract}

\begin{keywords}
Markov chain Monte Carlo;
Parallel computing; 
Log-concave distributions; 
Mixing time; 
Spectral gap;
Zeroth-order optimization
\end{keywords}

\section{Introduction} \label{sec:1}
There is substantial interest in developing Markov chain Monte Carlo (MCMC) methods that can exploit parallel computation in their routines to improve efficiency \citep{rosenthal2000parallel,liu2000multiple,tjelmeland2004using,brockwell2006parallel,calderhead2014general,gagnon2023improving}.
Being MCMC algorithms sequential in their very own nature, this task presents in general non-trivial complications, which resulted in simply running different chains parallel \citep{rosenthal2000parallel} being arguably the most popular solution applied in practice. Besides being relatively straightforward to implement, this approach has the advantage of running essentially without any communication cost between different processors. Moreover, in problems where the burn-in of the chain (i.e., the time needed for each chain to reach convergence) is small compared to the total running time, it is successful in increasing the effective sample size of the simulation.
On the contrary, for complex and high-dimensional target densities, MCMC algorithms can require to run for many iterations before converging to the stationary distribution. In these scenarios, the benefits of using parallel chain MCMC methods instead of classical single-processor ones tend to be more limited, making these solutions less attractive, especially if the length of each independent chain is limited. For this reason, over the last two decades, an alternative line of research has focused on developing MCMC methods that use parallel computation to accelerate convergence. These algorithms incorporate parallel computation within a single Markov chain and often have a more complex structure than parallel chains MCMC methods. In the following, we refer to this class of techniques as \emph{in-step parallelization} methods.
 
In this paper, we consider the problem of using parallel computations to speed-up convergence in the zeroth-order framework. Namely, the setting where only the potential of the target density can be evaluated, but not its derivatives. Although zeroth-order algorithms are generally less efficient than first and second order solutions (i.e., methods that exploit gradient and Hessian information), this framework is receiving renewed attention, particularly in optimization literature. This is because, even though the important developments of automatic differentiation techniques \citep{baydin2018automatic,margossian2019review} have made gradient-based algorithms easier to implement, there are still practical situations where applying these techniques is not always straightforward either for theoretical or computational reasons. 
Examples discussed in the literature \citep{nesterov2017random,kozak2021stochastic} include: i) models whose likelihood is obtained with complex black-box computer programs which do not include the possibility of computing derivatives (at least not without a substantial and time-consuming effort of a human programmer) and ii) target distributions which are not everywhere differentiable. Moreover, even when backward-mode automatic differentiation can be used to obtain the entire gradient at the cost of no more than four target evaluations \citep{nesterov2017random}, this approach still requires storing multiple intermediate steps.
In complex models, this can result in memory overflow, making gradient-based methods challenging to implement in these contexts \citep{kozak2021stochastic}.

In this paper, we adopt a perspective inspired by recent advances in zeroth-order optimization \citep[see e.g.,][]{duchi2015optimal,nesterov2017random,kozak2021stochastic}. In this field, there is increasing evidence that zeroth-order algorithms designed to mimic first-order methods are more efficient than fully derivative-agnostic approaches. These methods generalize classical gradient descent by replacing the exact gradient with a zeroth-order stochastic estimator, typically constructed using finite-difference approximations of partial derivatives along random directions.

We investigate whether this approach can be effectively extended to sampling. Our contribution is twofold. First, we show that naive applications to sampling of strategies developed in optimization fail to yield efficient parallel sampling algorithms, as discussed in Section \ref{sec:3}. Second, we propose a simple and effective zeroth-order sampling method that overcomes these inefficiencies, presented in Section \ref{sec:4}. Our approach is designed for high-dimensional settings where the number of available parallel processors $m$ is strictly smaller than the dimension $d$ of the target distribution.

\subsection{Contributions of the paper}

The paper is organized as follows. Section \ref{sec:2} introduces the zeroth-order framework, outlines the main technical assumptions (Section \ref{sec:sgradient}), and reviews optimization strategies based on zeroth-order stochastic estimators of the gradient (Section \ref{sec:2:2}). Section \ref{sec:3} examines the possibility of directly extending the approaches described in Section \ref{sec:2} to sampling, both for unadjusted and Metropolis-adjusted methods. This is done by replacing the exact gradient with a zeroth-order estimate in the unadjusted Langevin algorithm (ULA) \citep{durmus2019high} and its Metropolis-adjusted counterpart (MALA) \citep{roberts1996exponential}. These approaches are shown to perform suboptimally with respect to their dependence on $d$ and $m$ (Sections \ref{sec:naive:sula} and \ref{sec:naive:smala}, respectively).

Building on the insights from the previous section, Section \ref{sec:4} presents our main methodological contribution: the zeroth-order random slice sampler. Instead of using target evaluations to construct an estimate of the full gradient, we propose to restrict, at each iteration, the proposal distribution to a random $m$-dimensional subspace. On this subspace, the gradient can be well approximated deterministically using finite differences. This simple modification resolves the inefficiencies identified in Section \ref{sec:3}. In Section \ref{sec:4:2}, we combine the random slice sampler with a zeroth-order version of Hamiltonian Monte Carlo (HMC) \citep{neal2011mcmc}. Section \ref{sec:4:3} studies the theoretical properties of the zeroth-order random slice sampler, proving that it can achieve a polynomial improvement in efficiency with respect to the number of parallel workers, thus substantially outperforming existing zeroth-order parallel MCMC methods. The empirical performance of the proposed methodology is illustrated on two Bayesian regression models in Section \ref{sec:5}. All proofs are provided in Appendix \ref{sec:app:a}. Code is available at \url{https://github.com/Francesco16p/ZO_Parallel_sampling}

\subsection{Previous work about in-step parallelization of MCMC}

When the focus is on sampling, and a single computing unit is available, the random-walk Metropolis (RWM) MCMC algorithm is arguably the gold-standard zeroth-order sampler. If correctly implemented, this algorithm requires only one evaluation of the potential function per iteration. In the case of multiple parallel processors, the situation is more varied. Many of the proposed in-step parallelization methods can be cast as generalizations of RWM belonging to the multi-proposal family \citep{liu2000multiple,tjelmeland2004using,calderhead2014general,craiu2007acceleration,bedard2012scaling,casarin2013interacting,martino2017issues,fontaine2022adaptive,gagnon2023improving,yang2023convergence}. These algorithms share a common structure for which, at every iteration, a set of multiple points is proposed as  possible next states of the chain
and, subsequently, one of them is selected with a probability that maintains the transition kernel invariant with respect to the target. Importantly, this construction allows for parallelization at each step of the algorithm. When these methods are properly designed, a higher number of proposed points corresponds to better performance \citep{gagnon2023improving}. However, while these techniques can perform well in low- or moderate-dimensional problems, recent theoretical investigations \citep{pozza2025fundamental, caprio2025analysis} highlight that, in high-dimensional settings, the improvement with the number of parallel processors $m$ is usually limited. Indeed, in important situations such as high-dimensional log-concave targets, the efficiency improvement of multiple-proposal MCMC is is at most poly-logarithmic in $m$ \citep{pozza2025fundamental}.

An alternative form of in-step parallelization, compared to multiple-proposal methods, is pre-fetching \citep{brockwell2006parallel}. This algorithm uses parallel workers to simulate an RWM chain ahead of time by evaluating all of its possible trajectories. Although this approach differs substantially from that employed in multiple-proposal MCMC, the number of target evaluations increases exponentially with the length of the chain. This implies that the improvement for pre-fetching is only logarithmic in the number of parallel workers \citep{brockwell2006parallel}. Consequently, the use of in-step parallelization to improve efficiency and obtain MCMC methods that substantially outperform the RWM algorithm has essentially remained an open problem.

Recently, \citet{grazzi2025picard} developed an in-step parallelization for RWM using Picard maps in which performance increases polynomially with the number of processors $m$. 
Here we also obtain a polynomial increase in $m$, by taking a different perspective inspired by advances in zeroth-order optimization.

\section{Zeroth-order parallel sampling} 
\label{sec:2}
\subsection{The abstract framework}\label{sec:sgradient}
Let $\mathcal{B}(\R^d)$ be the Borel $\sigma$-field on $\R^d$ and denote by $\sP(\R^d)$ the set of probability measures on $(\R^d,\mathcal{B}(\R^d))$. Let $\pi$ be a target probability measure in $\sP(\R^d)$ and, with a slight abuse of notation, let us denote with $\pi(x) = \exp( - U(x) )$, $x \in \R^d,$ its probability density function with respect to the Lebesgue measure. We study the problem of generating samples from $\pi$ under the following assumptions:
\begin{itemize}
\item[(i)] A zeroth-order oracle for $U$ is available, i.e., $U$ can be evaluated but not its derivatives. Note that the distinctive aspect of this framework does not concern the existence or non-existence of derivatives of $U$, but rather the possibility of explicitly evaluating them.
\item[(ii)] $m \ge 1$ processors can be used in parallel to evaluate $U$ at $m$ different locations.
\end{itemize}
In the following, $m$ evaluations of $U$ performed in parallel is referred to as a \emph{parallel round} \citep{anari2024fast}. For simplicity, communication costs among processors are assumed negligible, so that the computational cost of one parallel round equals that of a single sequential evaluation of $U$.

In both optimization and sampling, where applicable, gradient information can be used to accelerate convergence. For this reason, one possible strategy for deriving good zeroth-order algorithms is to implement solutions that mimic derivative-based approaches by numerically estimating the derivatives of $U$ using zeroth-order queries. 
In serial applications, these methods are in general more computationally demanding compared to native first-order algorithms \cite{kozak2021stochastic}. However, if multiple parallel workers are available, part of the computations in their routines can be done in parallel, speeding up the algorithm considerably.

We consider algorithms that make use of the finite-difference approximation of directional derivatives reported in Algorithm \ref{alg:numdev}.
\begin{algorithm}[t]
\caption{$m$-dimensional finite-difference directional derivatives }\label{alg:numdev}
\KwIn{$x \in \R^d$, $m  \in \N$ and $\epsilon > 0$} 
Sample $\bV \sim \nu$\\
Compute
\begin{equation} \label{eq:FinDiff}
     \frac{ U(x + \epsilon v^{(i)} ) - U(x)}{ \epsilon },\,\qquad \mathrm{for} \quad i = 1, \dots,m.
\end{equation}
\end{algorithm}
Given a point $x \in \R^d$, a matrix $(v^{(1)},\dots, v^{(m)})=\bV \in \R^{d \times m} $ containing $m$ directions of interest and a sufficiently small step-size $\epsilon>0$, Algorithm \ref{alg:numdev} relies on the definition of derivative as the limit of the incremental ratio to provide an estimate of the vector of partial derivatives of $U$ at $x$ with respect to directions $\bV$. Taking the limit as $\epsilon \to 0$, every incremental ratio in \eqref{eq:FinDiff} converges to the corresponding directional derivative $ \partial_{v^{(i)}} U(x)$ of $U$ at $x$ (with respect to direction $v^{(i)}$), i.e., $ \partial_{v^{(i)}} U(x) = \lim_{\epsilon \to 0} ( U(x + \epsilon v^{(i)}) - U(x) )/\epsilon$. 
In sufficiently regular optimization problems, the error introduced by a fixed $\epsilon$ in Algorithm \ref{alg:numdev} is typically negligible compared to the precision required to locate the optimum \citep{nesterov2017random}.
In sampling, this issue is arguably even less compelling, since derivative information is not needed to identify a specific optimum, but rather to guide the algorithm towards regions of high probability. Another important property of directional derivatives is that, when the gradient $\nabla U(x)$ exists, one has
$
\partial_{v^{(i)}} U(x) = \nabla U(x)^\top v^{(i)}.
$
That is, each directional derivative equals the inner product of the gradient of $U$ and the chosen direction. This identity greatly simplifies the theoretical analysis of methods relying on directional derivatives, as it separates aspects involving properties of $\nabla U$ from those depending intrinsically on how $\bV$ is generated. In the following, to use this approach, we ignore the approximation error in Algorithm \ref{alg:numdev}, and consider $\partial_{v^{(i)}} U(x)$ directly in place of its finite-difference approximation \eqref{eq:FinDiff} as done, for example, in \citet{nesterov2017random}. 

We conclude this section with a brief discussion on the amenability of Algorithm \ref{alg:numdev} to parallel computing. Assuming that $U(x)$ has been previously obtained, Algorithm \ref{alg:numdev} requires $m$ additional evaluations of the target function. However, if $m$ parallel workers are available, each operation in \eqref{eq:FinDiff} can be performed independently and in parallel thus resulting in a single parallel round. As a consequence, if communication and synchronization overheads are ignored, with $m$ processors, Algorithm \ref{alg:numdev} can be executed in essentially the same time as a single zeroth-order oracle call.

\subsection{Comparison with zeroth-order optimization} \label{sec:2:2}
Zeroth-order optimization has been studied extensively in the literature and particular emphasis has been placed on stochastic gradient algorithms \citep{nesterov2017random,kozak2021stochastic,kozak2023zeroth}. These methods use directional derivatives, possibly estimated with Algorithm \ref{alg:numdev}, to construct unbiased estimates of the gradient of $U$, which in turn are incorporated into gradient descent-type algorithms to find $x_* = \arg \min_{x \in \R^d} U(x)$. Given the many similarities that exist between sampling and optimization, in this section we briefly summarize important results of the latter. This is done with the aim of highlighting some essential properties that an efficient zeroth-order sampler should have, especially with respect to its algorithmic efficiency in terms of target dimension $d$ and number of parallel workers $m$. 

Given a starting point $x_0 \in \R^d$, an iteration $t \in \N$ of a generic zeroth-order stochastic gradient descent algorithm can be described by the two following steps:

\begin{equation}
\begin{aligned}
    &\text{1. Generate } \bV \in \R^{d\times m} \text{ from a probability distribution } \nu, \\
    &\text{2. Set } x_t = x_{t-1} - l_t \hat{\nabla}_{\bV} U(x_{t-1}),
\end{aligned}
\label{eq:zo_stoch_opt}
\end{equation}
where $\hat{\nabla}_{\bV} U(x) = c_{\nu} \bV \bV^{\top} \nabla U(x)$, for some $c_{\nu}\in \R$ such that $\E[\hat{\nabla}_{\bV} U(x)] = \nabla U(x)$ for $\bV\sim \nu$, is a zeroth-order stochastic estimator of the gradient and $l_t$ is a learning rate. Note that, in view of the equality
$
\partial_{v^{(i)}} U(x) = \nabla U(x)^\top v^{(i)},
$
estimators satisfying the above conditions for $\hat{\nabla}_{\bV} U(x)$ can be obtained (if one ignores the approximation error arising from choosing $\epsilon > 0$) by combining Algorithm \ref{alg:numdev} with an appropriate choice for $\nu$ and $c_{\nu}$ \citep[][Eq. 3]{kozak2021stochastic}.

For two positive functions $f,g : \R^{+} \to \R^{+}$ we say that $f = O(g)$ if there exists $x_0 \in \R^{+}$ and a positive constant $C$ such that $f(x) \leq C g(x) $ for any $x \geq x_0$. When $m = 1$, it is known that zeroth-order stochastic gradient methods require $O(\kappa d\log(1/\epsilon))$ iterations to achieve $\epsilon$-error for strongly convex and smooth functions with condition number $\kappa$, see Assumption \ref{cond:1} for definition. By contrast, under the same conditions, first-order methods require only $O(\kappa \log(1/\epsilon))$
iterations \citep[Ch. 5]{bach2024learning}, meaning that replacing the exact gradient with a $m = 1$ zeroth-order estimate introduces an additional cost factor of order $d$. 
When $m>1$ and $m\leq d$ the number of iterations reduces to $O(\kappa d/m \log(1/\epsilon))$  \citep[][Corollary 1]{kozak2021stochastic}, meaning that increasing the number of partial derivative evaluations per round can reduce convergence time by a factor of $m$. 
This fact has significant practical implications for both serial and parallel implementations of zeroth-order optimization methods. Indeed, the accuracy achieved after $t$ iterations using $m^*>1$ directions is equivalent to that obtained after running $m^* \times t$ iterations with $m = 1$. Thus, in serial implementations and for an equal number of target evaluations, different values of $m$ yield algorithms of comparable efficiency. Conversely, in parallel settings, if $m$ workers are available, one can achieve the same accuracy obtained by running $t$ iterations on a single processor in just $t/m$ parallel rounds.

This linear speed-up is substantially higher than the logarithmic one observed for common in-step parallel MCMC methods \citep{brockwell2006parallel,pozza2025fundamental}. In this paper, we aim to develop zeroth-order parallel MCMC methods whose improvement with respect to $m$ more closely resembles that observed in optimization, thus being at least polynomial in $m$.

We conclude this section by commenting in more detail on the various options regarding the generation of $\bV$. Common choices are independent $d$-dimensional Gaussians and independent uniform on the $d$-dimensional unit sphere (a random variable uniform on the $d$-dimensional unit sphere can always be obtained by sampling $Z \sim N_d(0,\I_d)$ and taking $\bV = Z/ \| Z \|$ \citep[][pg. 68]{vershynin2018high}). While computationally attractive, drawing $m$ independent directions leads to the undesirable property that $m \geq d$ directional derivatives do not exactly reconstruct the gradient since, in general, $\bV \bV^\top \neq \I_d$. Different solutions to overcome this issue can be considered.  In the following, we assume $\bV \in \V_m(\R^d)$, with $\V_m(\R^d) = \{ \bV \in \R^{d \times m}, \bV^{\top} \bV = \I_k \}$ denoting the $m$-dimensional Stiefel manifold on $\R^d$. Generating a sample from the uniform distribution $\nu = \nu_{\mathrm{unif}}$ on $\V_m(\R^d)$ can be done at cost $O(m^2d)$ via QR decomposition of $d \times d$ random matrix with i.i.d  $N_d(0,\I_d)$ distributed columns  \citep[][Sec. 4]{mezzadri2006generate}. Note that, up to a multiplicative correction, $\nu = \nu_{\mathrm{unif}}$ is equivalent to what has been assumed for example in \citet{kozak2021stochastic,kozak2023zeroth}. Alternatively, a less computationally intensive solution is to consider $\nu = \nu_{\mathrm{E}}$ which takes at each iteration an $m$-dimensional random subset of the canonical basis $\mathbf{E} = (e^{(1)}, \dots, e^{(d)}) $ where $e^{(i)}_i = 1$ and $e^{(i)}_j = 0$ for $i \neq j,\, j,i = 1, \dots,d$. In \eqref{eq:zo_stoch_opt}, this choice for $\bV$ is analogous to randomized coordinate descent \citep[][pg. 194]{aggarwal2020linear}, while in sampling it corresponds to Gibbs-type MCMC algorithms \citep{casella1992explaining,gelfand2000gibbs}. Note that, when $m = d$, $c_{\nu} \bV \bV^{\top} \nabla U(x) = \nabla U(x)$ for both $\nu = \nu_{\mathrm{unif}}$ and $\nu = \nu_{\mathrm{E}}$ if $c_{\nu} = d/m$. This means that, ignoring the error due to finite difference approximations, Algorithm \ref{alg:numdev} exactly reconstructs the gradient.

\section{Naive plug-in approaches} \label{sec:3}
\subsection{Unadjusted methods} \label{sec:naive:sula}
Once an estimate of $\nabla U(x)$ is available, a natural question is how to exploit it for designing efficient sampling algorithms. One possibility is to replace $\nabla U(x)$ with $\hat \nabla_{\bV} U(x)$ within any gradient-based sampler. By generating at each iteration a different realization of $\bV$, this strategy can, in principle, be adapted to both unadjusted \cite{durmus2019high} and adjusted \citep[][]{roberts1996exponential,neal2011mcmc,livingstone2022barker} gradient-based MCMC methods. 

In this section, we deal with the unadjusted framework by considering the case in which $\hat \nabla_{\bV} U(x)$ is incorporated in the unadjusted Langevin algorithm (ULA) \citep{roberts2002langevin}. In Algorithm \ref{alg:sula}, at iteration $t$, a set of $m$ random directions $\bV$ is sampled. Conditioned on $\bV$ a step of ULA is performed with the exact gradient replaced by its stochastic estimate obtained as in \eqref{eq:zo_stoch_opt}.

\begin{algorithm}
\caption{Zeroth-order unadjusted Langevin algorithm}\label{alg:sula}
\KwIn{$T \in \N$, $\gamma>0$, $c_{\nu} \in \R$, $\mu \in \mathcal{P}(\R^d)$ and $\nu$ distribution on $\V_m(\R^d)$.}
\vspace{3pt}
$X_0 \sim \mu$ \\
\For{$t = 1,\dots,T$ 
}{
Sample $\bV \sim \nu$ and $Z \sim N_d(0, \I_d)$ \\
Set $X_{t} =X_{t-1} - \gamma \hat{\nabla}_{\bV} U(X_{t-1}) + \sqrt{2 \gamma } Z $ with $\hat{\nabla}_{\bV} U(X_{t-1})= c_{\nu}\bV \bV^{\top} \nabla U(X_{t-1})$
}
\vspace{3pt}
\KwOut{Markov chain trajectory $(X_0,X_1,\dots, X_T)$.}
\end{algorithm}

Theoretical analysis of unadjusted MCMC algorithms typically involves addressing two distinct aspects \citep{durmus2019high}. First, given a particular metric of interest (often the Wasserstein distance of order 2 or the total variation distance), it is necessary to determine the speed at which the process converges to its stationary distribution, $\pi_{\gamma}$. Secondly, the distance between $\pi_{\gamma}$ and $\pi$ must be quantified. In the following we argue that, at least at a theoretical level, Algorithm \ref{alg:sula} does not represent a particularly appealing way of employing parallel computing resources. These inefficiencies stem mainly from how the invariant distribution of the algorithm approximates the target rather than how fast the algorithm reaches stationarity.

For two probability measures $\mu,\eta \in \sP(\R^d)$, a probability measure $\xi$ on $(\R^d \times \R^d, \mathcal{B}(\R^d) \times  \mathcal{B}(\R^d) )$ is a coupling of $\mu$ and $\eta$ if for every measurable set $A \subseteq \R^d$, $\xi(A \times \R^d) = \mu(A)$ and $\xi(\R^d \times A) = \eta(A)$. The set of all couplings of $\mu$ and $\eta$ is denoted with $\Xi(\mu,\eta)$. Finally, we define the Wasserstein distance of order 2 as
\begin{equation*}
    W_2(\mu, \eta) = \sqrt{ \inf_{\xi \in \Xi(\mu,\eta) } \int_{\R^d \times \R^d} \|x-y\|^2 \xi(dx,dy) }. 
\end{equation*}

To investigates the contraction properties, in Wasserstein distance of order 2, of Algorithm \ref{alg:sula} toward its stationary distribution we require two additional assumptions, the first on $\pi$ and the second on $\nu$. 
\begin{assumption} \label{cond:1}
$\pi \in \mathcal{P}(\R^d)$ has density $\exp(-U)$ where $U \, : \, \R^d \to \R$ is 
\begin{enumerate}
    \item $L$-smooth, meaning that $\nabla U : \, \R^d \to \R^d$ is $L$-Lipschitz.
    \item $\lambda$-convex, meaning that $U(x) - (\lambda/2)\|x\|^2$ is convex.
    \item Twice continuously differentiable.
\end{enumerate}
In the following we denote with $\kappa = L/\lambda$ is the condition number of $\pi$.
\end{assumption}
\begin{assumption} \label{cond:2}
$c_\nu=d/m$ and $\nu$ is such that $(d/m) \E [\bV \bV^\top] = \I_d$ for $\bV\sim \nu$.
\end{assumption}

The convexity and smoothness conditions in Assumption \ref{cond:1} are common in theoretical analysis of MCMC methods \citep{chewi2024log}. They allow for explicit quantitative bounds while including a large and non-trivial class of distributions. Assumption \ref{cond:2} is more specific to the problem at hand and it characterizes the sampling mechanism of $\bV$. Conditions similar to Assumption \ref{cond:2} can be found in \citet[][Sec. 8]{ascolani2024entropy} and in the the stochastic optimization literature \citep{kozak2021stochastic}. Distributions satisfying Assumption \ref{cond:2} are the uniform distribution on $\bV_m(\R^d)$ and that obtained by sampling uniformly and without replacement $m$ elements from the natural basis of $\R^d$. 

\begin{proposition} \label{prop:contr:sula}
Let $P^{SULA}$ be the transition kernel of Algorithm \ref{alg:sula}. Then, under Assumptions \ref{cond:1} and \ref{cond:2}, for every $x,y\in\R^d$ and $ \gamma\leq m/(Ld)$ it holds that
\begin{equation} \label{eq:contr:sula}
    W_2(\delta_x P^{SULA},\delta_y P^{SULA})
    \leq
    \left(
       1-  \gamma  \lambda 
    \right)^{1/2}
    \|x-y\|  \leq \exp \left( - \frac{1}{2} \gamma \lambda \right) \|x-y\|,
\end{equation}
where $\delta_x,\delta_y$ denote the Dirac measure on $x$  and $y$, respectively.
\end{proposition}
For any $\epsilon>0$, Proposition \ref{prop:contr:sula} implies that, taking $ \gamma= m/(Ld)$,  Algorithm \ref{alg:sula} requires $n_* = O(\kappa d/m \log(1/\epsilon))$ parallel rounds to produce a sample from a distribution $\epsilon$-close to its own stationary distribution $\pi_{\gamma}$. 
This rate matches the best one obtained in analogous optimization contexts \citep[][Corollary 1]{kozak2021stochastic}, meaning that in this setting there is no slow-down when moving from optimization to sampling in terms of speed of convergence. 
However, similarly to the case of standard ULA, the stationary distribution $\pi_\gamma$ of Algorithm \ref{alg:sula} can differ significantly from the target $\pi$ if the step-size $\gamma$ is not suitably scaled. 
In particular, the required conditions have been studied in a similar context in \citet{roy2022stochastic}. Their framework includes the case where $\pi$ is evaluated only up to stochastic error, but the results therein apply also to our context, where such error is zero. Their findings suggest that schemes analogous to Algorithm \ref{alg:sula} require $O(d^2/m)$ parallel rounds to sample from an accurate approximation to the target $\pi$, see Theorem 2.1 therein. In view of Proposition \ref{prop:contr:sula}, this result primarily stems from the need to take a small stepsize, specifically $\gamma=O(m/d^2)$, to ensure small bias in the stationary distribution $\pi_\gamma$. 
Since, under similar conditions, standard RWM requires only $O(\kappa d\log(1/\epsilon))$ sequential zeroth-order evaluations of $U$ for accurate sampling \citep[][Theorem 49]{andrieu2024explicit}, it follows that Algorithm \ref{alg:sula} represents, at least theoretically, a suboptimal use of parallel computing for sampling in this zeroth-order context.

\subsection{Naive approach to Metropolis-adjusted methods} \label{sec:naive:smala}

Section \ref{sec:naive:sula} highlights that the inefficiencies found in Algorithm \ref{alg:sula} are mainly attributable to the discrepancy from its stationary distribution and the target. As a result, a way to possibly avoid this issue is to consider instead a Metropolis-adjusted algorithm. In this section, we further investigate this possibility.

Algorithm \ref{alg:smala} illustrates how $\hat \nabla_{\bV} U(x)$ can be incorporated in the Metropolis-adjusted Langevin algorithm (MALA). At iteration $t$, a set of $m$ random directions $\bV$ is sampled. Conditioned on $\bV$, a step of MALA is performed with the exact gradients replaced by stochastic estimates obtained as in \eqref{eq:zo_stoch_opt}.

\begin{algorithm}[H]
\caption{Naive zeroth-order Metropolis-adjusted Langevin algorithm }\label{alg:smala}
\KwIn{$T \in \N$, $\sigma^2>0$, $\mu \in \mathcal{P}(\R^d)$ and $\nu$ distribution on $\V_m(\R^d)$}
\vspace{3pt}
$X_0 \sim \mu$ \\
\For{$t = 1,\dots,T$ 
}{
Given $X_t = x$, sample:\\
\hspace{25pt} $ \bV \sim \nu $, \\
\hspace{25pt} $y \sim Q_{\bV}(x,\cdot)$ where $ Q_{\bV}(x,\cdot) = N_d(x - (\sigma^2/2) \hat{\nabla}_{\bV}U(x), \sigma^2 \I_d )$.\\
Set $X_{t+1} = y $ with probability 
$$
\min\left(
1,
\frac{\pi(dy) Q_{\bV}(y,dx)}{\pi(dx) Q_{\bV}(x,dy) }
\right),
$$
$X_{t+1} = X_t $ otherwise.
}
\vspace{3pt}
\KwOut{Markov chain trajectory $(X_0,X_1,\dots, X_T)$}
\end{algorithm}

Equivalently, the transition kernel associated to Algorithm \ref{alg:smala} can be written as
\begin{equation} \label{eq:kernel:smala}
    P^{\scriptscriptstyle SMALA}_{\sigma}(x,A) = \int_{\R^d \times \V(\R^d)}
    \mathbbm{1}_{A}(y)
    Q_{\bV}(x,dy)
    \min\left(
               1,
               \frac{\pi(dy) Q_{\bV}(y,dx)
               }
               {
               \pi(dx) Q_{\bV}(x,dy)
               }
\right) \nu(d\bV),
\end{equation}
which can be easily proven to be $\pi$-invariant as shown in Proposition \ref{prop:pi:inv:nsmala} below.
\begin{proposition} \label{prop:pi:inv:nsmala}
    $P^{\scriptscriptstyle SMALA}_{\sigma}$ in \eqref{eq:kernel:smala} is $\pi$-invariant.
\end{proposition}
Algorithm \ref{alg:smala} represents a fairly straightforward way of including gradient estimates from \eqref{eq:zo_stoch_opt} into Metropolis-adjusted gradient-based MCMC. In this section we formally demonstrate that, as for the unadjusted case discussed in Section \ref{sec:naive:sula}, such an approach is inefficient, unless $m$ is close to $d$. In the following, $\bV^{\perp}$ denotes an arbitrary $d \times (d-m)$ matrix with columns forming an orthonormal basis of the orthogonal complement of $\mathrm{span}(\bV)$.
\begin{lemma} \label{lemma:prop:smala}
Let $\bV\in\V_m(\R^d)$, $x\in\R^d$ and 
$y \sim Q_{\bV}(x,\cdot)$, with $Q_{\bV}$ as in Algorithm \ref{alg:smala}. Then $s = \bV^{\top} y$ and $s^{\perp} = (\bV^{\perp})^{\top} y$ are independent random variables, for fixed $x$ and $\bV$, and
\begin{align}
    s &\sim N_{m}\big( \bV^\top x - (\sigma^2/2) \bV^{\top} \hat{\nabla}_{\bV} U(x) , \sigma^2 \bV^{\top} \bV), \label{eq:prop:s:smala} \\
    s^{\perp} &\sim  N_{d-m}\big( (\bV^{\perp})^{\top} x , \sigma^2 \I_{d-m}). \label{eq:prop:sp:smala}
\end{align}
\end{lemma}

The decomposition of $Q_{\bV}(x,\cdot)$ in Lemma \ref{lemma:prop:smala} shows that the proposal distribution of Algorithm \ref{alg:smala} behaves differently on the subspace spanned by $\bV$ and on its orthogonal complement, a distinction that has potentially important consequences for its global behavior. As $m$ increases, Algorithm \ref{alg:smala} is expected to become equivalent to a preconditioned version of MALA \citep{roberts2002langevin,girolami2011riemann}. On the contrary, if $d-m$ is not small the random walk component in \eqref{eq:prop:sp:smala} is likely to drive the performance of the algorithm. As a result, it is crucial to understand how the interplay between  \eqref{eq:prop:s:smala} and \eqref{eq:prop:sp:smala} affects the overall efficiency of the algorithm.

Theorem \ref{theo:gap:smala} below further investigates this issue by studying how the right spectral gap of $P_{\sigma}^{MALA}$ is influenced by $m$. Recall that for a $\pi$-reversible transition kernel $P$ the right spectral gap is defined as 
\begin{equation} \label{eq:def:gap}
    Gap_{P} = \inf_{f \in L^2_{\pi}} \frac{ \int_{\R^{2 d}}
    \lbrace 
    f(x) -f(y) 
    \rbrace^2
    \pi(dx)
    P(x, dy)
    }
    { 2 \V_{\pi}(f)},
\end{equation}
where $\V_{\pi}(f) = \int_{\R^d} f(x)^2 \pi(dx) - \lbrace \int_{\R^d} f(x) \pi(dx) \rbrace^2$ and $L^2_{\pi}$ is the class of all functions $f \, : \, \R^d \to \R$ such that $\V_{\pi}(f) < \infty$. The right spectral gap quantifies the difference between the two largest positive eigenvalues associated with $P$ and is equivalent to the spectral gap for positive-definite kernels. While for non-positive kernels the right spectral gap and the spectral gap are generally different, any transition kernel $P$ can always be linked to a positive-definite kernel by considering its lazy version $P_{\textsc{lazy}} = (P + I)/2$. In MCMC theory, the spectral gap is often considered as a standard measure of efficiency for $\pi$-reversible kernels \citep[Sec.12]{levin2017markov}, with higher spectral gaps being associated with better performance.
\begin{theorem} \label{theo:gap:smala}
     Under Assumptions \ref{cond:1}-\ref{cond:2}, if 
     $\nu$ is such that any column of $\bV$ is an element of the canonical basis $\mathbf{E}$, then
    \begin{equation} \label{eq:theo:gap:smala:result2}
    \sup_{\sigma>0} 
    Gap_{P^{\scriptscriptstyle SMALA}_{\sigma}} = O\Big( \frac{\log(d-m)^2 }{d-m} \Big) = O\Big( \frac{\log(d)^2 }{d} \Big),
    \end{equation}
    as $d\to\infty$ and $m< c_0 d$ for some $c_0 \in (0,1)$.
\end{theorem}
The proof of the theorem is based on the decomposition for the proposal distribution of $P^{\scriptscriptstyle SMALA}_{\sigma}$ given in Lemma \ref{lemma:prop:smala} combined with, non-asymptotic, arguments developed for RWM in \citet{andrieu2024explicit} and \citet{pozza2025fundamental}. 
Theorem \ref{theo:gap:smala} implies that, when $P^{\scriptscriptstyle SMALA}_{\sigma}$ is optimally tuned, the right spectral gap of Algorithm \ref{alg:smala} is at most of order $1/(d-m)$ up to a poly-logarithmic term. Under similar conditions, \citet{andrieu2024explicit} provide a lower bound of order $1/d$ for Gaussian random-walk Metropolis on $d$-dimensional targets. As a result, when $m$ is small or moderate compared to $d$, Algorithm \ref{alg:smala} behaves essentially like a RWM algorithm, thus providing an inefficient use of parallel computation.

Theorem \ref{theo:gap:smala} is illustrated in Figure~\ref{fig:gain_naive}. The figure reports, for a 200-dimensional Bayesian logistic regression example given in Section \ref{sec:5:1} and for different values of $m$, the gain over RWM obtained by Algorithm~\ref{alg:smala} (NAIVE-MALA) and compares it with that of the random-slice MALA (RS-MALA) algorithm, which is defined in Section~\ref{sec:4} below as a way to overcome the limitations of the methods analyzed in Section~\ref{sec:3}. For NAIVE-MALA, even when $m \simeq 100$, the improvement over RWM is small ($\leq 5$), while for the same values of $m$ RS-MALA is already 30–35 times more efficient than RWM. Figure~\ref{fig:gain_naive} also shows that the $O(1/(d-m))$ upper bound in Theorem~\ref{theo:gap:smala} describes well the behavior observed in practice.
See \eqref{eq:esjd} below for a formal definition of the measure of improvement displayed in Figure~\ref{fig:gain_naive}.

\begin{figure}[htbp]
  \centering
    \includegraphics[width=0.6\linewidth]{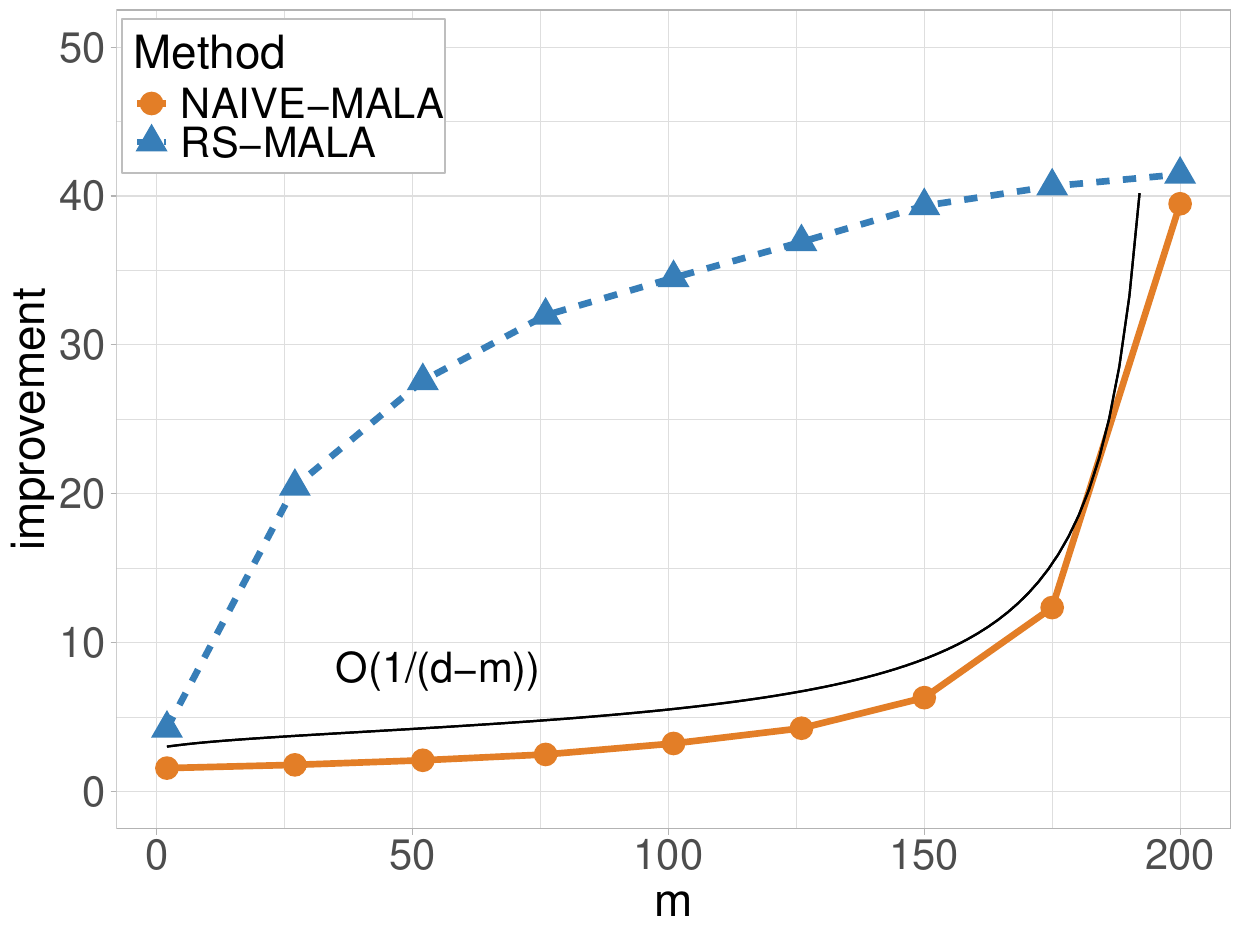}
\caption{Relative improvement over RWM, measured by the estimated average ESJD (as defined in \eqref{eq:esjd}), for NAIVE-MALA and for RS-MALA in the 200-dimensional logistic regression example of Section~\ref{sec:5:1}.}
\label{fig:gain_naive}
\end{figure}

\section{A random slice approach} \label{sec:4}
\subsection{The general methodology} \label{sec:4:1}

The main message of Section~\ref{sec:3} is that, while successful in optimization, direct plug-in of zeroth-order stochastic gradient estimators into MCMC may lead to algorithms with poor performance as the number of parallel workers increases.
The specific reasons can vary across algorithms, but they all stem from the lack of robustness of the methods presented in Section~\ref{sec:3} to the stochastic error in $\hat{\nabla}_{\bV} U(x)$. As shown in Section~\ref{sec:naive:smala}, in Metropolis-adjusted schemes this issue arises because $\hat{\nabla}_{\bV} U$ lies entirely in the $m$-dimensional subspace $\mathrm{span}(\bV)$. As a result, the proposal carries no information about the behavior of $U$ on the orthogonal complement $\mathrm{span}(\bV^{\perp})$ and, on this subspace, it behaves like a simple RWM.

Motivated by the insights of Section~\ref{sec:3}, in this section we propose to restrict the proposal distribution to act only on $\mathrm{span}(\bV)$, thereby removing the random-walk component entirely. In practice, this is equivalent to correlating the noise of the Metropolis-Hastings proposal with that of the stochastic gradient estimator, so that randomness is injected only in directions where the gradient is reliably estimated. 
For example, to correct the zeroth-order MALA described in Section \ref{sec:naive:smala}, this reduces to using the following proposal distribution:
\begin{equation*}
    y = x - (\sigma^2/2) \bV \bV^{\top}\nabla U(x) + \sigma \bV z, \qquad \text{with} \qquad z\sim N_{m}(0, \I_m),
\end{equation*}
where the $m$-dimensional noise $z$ is projected into $\mathrm{span}(\bV^{\perp})$ through the term $\sigma \bV z$.

More generally, this random slice approach can be combined with any gradient-based Metropolis-Hastings proposal distributions other than MALA, as done below for HMC, or any other high-accuracy gradient-based sampler.
The general methodology is described in Algorithm \ref{alg:mdim:rps}, where we use the following notation. 
Given $x \sim \pi$, let $s$ and $s^{\perp}$ be the random variables defined as $s^{\perp} = (\bV^{\perp})^{\top} x $ and $s = \bV^{\top} x$, respectively. Define the conditional distribution of $s$ given $s^{\perp}$ and $\bV$ as $\smash{\pi_{\bV,s^{\perp}}(s) \propto \exp\lbrace - U_{\bV, s^{\perp}}(s)\rbrace}$  with $ U_{\bV, s^{\perp}}(s) = U(\bV^{\perp} s^{\perp} + \bV s)$ for $s \in \R^m$. Moreover, let $\eta_{\bV,x}$ denote the density of a generic approximation of $\pi_{\bV,s^{\perp}}$, where the subscript $x$ in $\eta_{\bV,x}$ is adopted to highlight the possible dependence of the proposal on the current state of the chain.
\begin{algorithm}
\caption{$m$-dimensional random slice sampler}\label{alg:mdim:rps}
\KwIn{$T\in \N$, $\mu$, $\nu$, $\lbrace \eta_{\bV,x} \in \mathcal{P}(\R^m)\rbrace_{\bV \in \V_m(\R^d), \, x \in \R^{d}} $  }
\vspace{2pt}
$X_0 \sim \mu$\\
\For{$t = 1,\dots,T$ 
}{
Given $X_{t-1} = x$, sample $ \bV \sim \nu $ and set $s^{\perp} = (\bV^{\perp})^{\top} x $\\
\vspace{2pt}
Sample $s' \sim \eta_{\bV,x} $\\ 
\vspace{2pt}
Set $X_{t} = \bV^{\perp}s^{\perp} + \bV s' $ \
}
\vspace{3pt}
\KwOut{Markov chain trajectory $(X_0,X_1,\dots, X_T)$.}
\end{algorithm}

Given the current state $X_t = x$, the following two steps are performed: (i) a random set of $m$ directions $\bV$ is sampled from $\nu$ and (ii) only the projection of $x$ onto $\mathrm{span}(\bV)$ is updated according to $\smash{\eta_{\bV,x}}$.
Equivalently, the transition kernel of Algorithm \ref{alg:mdim:rps} can be expressed as 
\begin{align}\label{eq:kernel:bbhar}
    P(x,A) \, = 
    \,\int_{V_m(\R^d)}\int_{\R^m}
    \mathbbm{1}_A( \bV^{\perp} (\bV^{\perp})^{\top} x +\bV s' ) \eta_{\bV, x}(d s') \nu(d \bV),
\end{align}
for $x\in\R^d$ and $A\subseteq \R^d\backslash \{x\}$.

Algorithm~\ref{alg:mdim:rps} requires specifying the form of $\smash{\eta_{\bV,x}}$. In practice, we propose to take $\smash{\eta_{\bV,x}}$ as an update from a gradient-based sampler targeting $\smash{\pi_{\bV, s^{\perp}}}$. With this choice, Algorithm~\ref{alg:mdim:rps} can be interpreted as a gradient-based conditional update within an $m$-dimensional hit-and-run sampler \citep{chen1996general}. Crucially, since $\smash{\pi_{\bV, s^{\perp}}}$ is $m$-dimensional, its gradient can be numerically estimated using Algorithm \ref{alg:numdev} with a single parallel round of $m$ zeroth-order evaluations per iteration. Below, we show that the Markov kernel defined in~\eqref{eq:kernel:bbhar} exhibits an improved dependence on $m$ compared with the methods considered in the previous sections.

\subsection{Proposed practical methodology: Random-slice zeroth-order HMC} \label{sec:4:2}

The performance of any algorithm expressible as in \eqref{eq:kernel:bbhar} depends on the quality of the conditional proposal distribution $\eta_{\bV,x}$. In this section, we discuss a solution that we found to be particularly effective in practice, namely combining Algorithm \ref{alg:mdim:rps} with a zeroth-order version of HMC \citep{neal2011mcmc} on the $m$-dimensional slice.
Specifically, given $X_t = x$, each iteration of Algorithm \ref{alg:zohmc} consists of the following steps: (i) a new set of directions $\bV$ is sampled from a proposal distribution $\nu$; (ii) an HMC step with target distribution $\smash{\pi_{\bV, s^{\perp}}}$ is performed with the gradient of $\smash{U_{\bV, s^{\perp}}}$, computed via finite differences as in Algorithm~\ref{alg:numdev}.
\begin{algorithm}
\label{alg:zohmc}
\renewcommand{\baselinestretch}{1.4}\selectfont
\caption{Random-slice zeroth-order Hamiltonian Monte Carlo }\label{alg:rshmc}
\KwIn{$X_0 \in \R^d$, $\epsilon,\gamma > 0$, $m,L,T \in \N$, $U$ target's potential and $\nu$ distribution on $\V_m(\R^d)$.} 
\For{$t = 1,\dots,T$ 
}{
Sample: $\bV \sim \nu$ and $k \sim N_m(0, \I_m)$ \\ 
Given $X_{t-1} = x$ set: \\
\hspace{17pt}  $k'_0 = k - \frac{\gamma}{2} \sum_{i=1}^m \Big[ \frac{ U(x + \epsilon v^{(i)}) - U(x)}{\epsilon}\Big] v^{(i)}$  \\ 
\hspace{17pt} $s_0 = 0 \in \R^m$ \\
\vspace{2pt}
\For{$l = 1,\dots,L$ 
}{
  $s_{l} = s_{l-1} + \gamma k'_{l-1}$ \\
 \If{$l < L$}
    {
    $k'_{l} = k_{l-1} - \gamma \sum_{i=1}^m \Big[ \frac{ U(x + \bV s_l + \epsilon v^{(i)}) - U(x + \bV s_l)}{\epsilon}\Big] v^{(i)}$
    }
}
\vspace{5pt}
$k'_{L} = - k_{L-1} + \frac{\gamma}{2} \sum_{i=1}^m \Big[ \frac{ U(x + \bV s_L + \epsilon v^{(i)}) - U(x + \bV s_L)}{\epsilon}\Big] v^{(i)}$\\
Set $X_{t} = x + \bV s_L $ with probability 
$$
\min\left(
1,
\exp\Big( U(x) - U(x + \bV s_L) + 0.5( \|k\|^2 - \|k'_{L}\|^2) \Big)\right),
$$
or $X_{t} = X_{t-1} $ otherwise.
}
\vspace{5pt}
\KwOut{Markov chain trajectory $(X_0,X_1,\dots, X_T)$.}
\end{algorithm}

Some aspects of Algorithm~\ref{alg:rshmc} deserve further discussion. The support of $\nu$ is $\V_m(\R^d)$. Therefore, in principle, many different choices are possible. In our opinion, the most practically appealing are the uniform distribution on $\V_m(\R^d)$ (denoted with $\nu = \nu_{\mathrm{unif}}$) and $\nu = \nu_{\mathrm{E}}$, which selects $m$ directions at random from the canonical basis of $\R^d$ (see also Section \ref{sec:2}). The former may offer theoretical benefits, such as invariance under rotational reparameterizations, but the latter allows for significantly cheaper sampling. For this reason, $\nu_{\mathrm{E}}$ is our preferred choice in practice. When $\nu = \nu_{\mathrm{E}}$, Algorithm~\ref{alg:rshmc} may be interpreted as a zeroth-order variant of an HMC-within-Gibbs sampler.
Another important point is that each iteration of Algorithm~\ref{alg:rshmc} does not perform an exact HMC update for $\smash{\pi_{\bV,s^{\perp}}}$, as the gradient of $\smash{U_{\bV, s^{\perp} }}$ is only numerically approximated. Lemma~\ref{prop:rshmc} below shows that the algorithm remains $\pi$-invariant.
\begin{proposition} \label{prop:rshmc}
The Markov transition kernel $P^{\textsc{rs-hmc}}_{m,L}$ of Algorithm~\ref{alg:rshmc} is $\pi$-invariant.
\end{proposition}
The proof of Proposition~\ref{prop:rshmc} follows from the deterministic proposal example in \citet{tierney1998note}, combined with the fact that the map $(s_L, k'_L) = T_L(s_0, k)$ used in Algorithm~\ref{alg:rshmc} is an involution and that the determinant of its Jacobian equals one. 

As for the standard HMC algorithm, the efficiency of Algorithm~\ref{alg:zohmc} may also benefit from a suitable preconditioning. Given a preconditioning matrix $A \in \R^{d \times d}$, a simple and effective approach is to reparameterize $y_t = A x_t$ and perform one iteration of Algorithm \ref{alg:zohmc} in the transformed space to compute $y_{t+1}$. The resulting point is then mapped back to the original space by taking
$
x_{t+1} = A^{-1}y_{t+1}.
$  
The empirical performance of Algorithm~\ref{alg:rshmc} is studied through two simulation studies in Section \ref{sec:5}. 

\subsection{A theoretical bound on the complexity of Algorithm \ref{alg:mdim:rps}} \label{sec:4:3}
We conclude this section by showing that, when combined with suitable proposal distributions, the performance of Algorithm \ref{alg:mdim:rps} can increase polynomially with the number of parallel processors. This indicates that the proposed method can offer a substantially more efficient alternative in practice to the approaches previously discussed in the literature and to those analyzed in Section \ref{sec:3}, as suggested by Figure \ref{fig:gain_naive} above.

In our analysis, two distinct aspects must be taken into account:
\begin{enumerate}
\item As for any hit-and-run-type algorithm, each iteration of Algorithm~\ref{alg:mdim:rps} acts only on an $m$-dimensional subspace of $\R^d$. The theoretical analysis of such algorithms requires proof techniques that differ from those used for methods operating on the full space. Under Assumption~\ref{cond:1}, convergence properties in terms of Kullback-Leibler contraction for MCMC methods with Gibbs-type and hit-and-run-type updates have been studied by \citet{ascolani2024entropy}. We build on these results by adapting them to the zeroth-order framework considered here.
\item The algorithm is strongly influenced by the properties of the proposal density $\eta_{\bV,x}$, both in terms of approximation quality (relative to $\smash{\pi_{\bV,s^{\perp}}}$) and computational cost. Consequently, it is important to disentangle effects that are intrinsic to the structure described in \eqref{eq:kernel:bbhar} from those that arise from the specific form of $\eta_{\bV,x}$. We deal with this fact by assuming that, for every $\delta>0$, $ x \in \R^d$ and $\bV \in \V_m(\R^d)$, the proposal $\eta_{\bV,x}$ satisfies 
$\mathrm{KL}(\eta_{\bV,x} \mid \pi_{\bV,s^{\perp} }) \leq \delta$, where
$$\mathrm{KL}(\eta_{\bV,x} \mid \pi_{\bV,s^{\perp}})
=
\int \log \left( 
\frac{ \pi_{\bV,s^{\perp}}}
{\eta_{\bV,x}}(s)
\right)
\eta_{\bV,x}(ds),$$
is the Kullback–Leibler divergence of $\pi_{\bV,s^{\perp}}$ from $\eta_{\bV,x}$. The existence of such a $\smash{\eta_{\bV,x}}$ for strongly log-concave and log-smooth target distributions on $\R^m$ and its computational cost is discussed in the last part of this section.
\end{enumerate}
Under Assumptions \ref{cond:1}-\ref{cond:2}, the efficiency of Algorithm \ref{alg:mdim:rps} in terms of Kullback-Leibler contraction from $\pi$ is studied in Theorem \ref{theo:appr:hr}. In the following, $P^t$ denotes the transition kernel obtained by applying $P$ (as defined in \eqref{eq:kernel:bbhar}) for $t$ times while, for a generic probability distribution $\mu \in \mathcal{P}(\R^d),$ $\mu P^t(A) = \int \mu(dy)P^t(y,A)$, for $A\subseteq \R^d$.
\begin{theorem} \label{theo:appr:hr}
Assume that for every $\delta>0$, $\bV \in \V_{m}(\R^d)$ and $x \in \R^d$, 
\begin{equation} \label{eq:cond:theorem:appr:hr}
  \mathrm{KL}(\eta_{\bV,x} \mid \pi_{\bV,s^{\perp} }) \leq \delta,  
\end{equation}
and that Assumptions \ref{cond:1}-\ref{cond:2} hold. Then, the transition kernel $P$ defined in \eqref{eq:kernel:bbhar} satisfies
    \begin{equation} \label{eq:result1:theorem:appr:hr}
        \kl(\mu P^t \mid \pi ) \leq \Big( 1 - \frac{1}{\kappa} \frac{m}{d} \Big)^t \kl(\mu \mid \pi ) + \frac{\delta \kappa d }{m} ,
    \end{equation}
    for all $t\in\mathbb{N}$.
Thus if $ \delta \leq \epsilon m /(2\kappa d)$ and $t\geq (\kappa d 
/m) \log(2\kl(\mu\mid\pi)/\epsilon)$ one has 
$$\kl(\mu P^t \mid \pi ) \leq \epsilon\,.$$ 
\end{theorem}

For every fixed $\epsilon$, Theorem \ref{theo:appr:hr} implies that Algorithm \ref{alg:mdim:rps} requires $O(\kappa d/m)$ iterations to obtain a sample from a distribution that is $\epsilon$-close to $\pi$ in Kullback--Leibler divergence. The improvement is therefore linear in $m$. This, however, does not account for the computational cost of $\eta_{\bV,x}$. 

Under Assumption \ref{cond:1} and in a first-order framework (i.e., when it is assumed the existence of an oracle providing information both on the target and its gradient),
one can construct a distribution $\smash{\eta_{\bV,x}}$ such that $\smash{\mathrm{KL}(\eta_{\bV,x} \mid \pi_{\bV,s^{\perp}}) \leq \delta}$, with $O\left(\kappa^{3/2} m^{1/2} \log^3(1/\delta)\right)$
sequential oracle calls. This, for example, by following the approach of \citet[Theorem D.1]{altschuler2024faster}. A first-order assumption does not meet the zeroth-order setting we consider. However, if one assumes that the approximation error in Algorithm \ref{alg:numdev} is negligible and that $m$ parallel processors are available, the gradient can be reconstructed in a single parallel round. In view of these considerations, we introduce Assumption \ref{cond:3}.
\begin{assumption} \label{cond:3}
For every $\delta>0$, $\bV \in \V_m(\R^d)$ and $x \in \R^{d} $, there exists an algorithm that uses
\begin{equation*}
    N(m,\kappa,\delta) = O\left(\kappa^{3/2} m^{1/2} \log^3(1/\delta)\right),
\end{equation*}
parallel rounds, each one consisting in $m$ parallel evaluations of $U$, to sample from a distribution $\eta_{\bV,x} \in \mathcal{P}(\R^m)$ such that $\mathrm{KL}(\eta_{\bV,x} \mid \pi_{\bV,s^{\perp} }) \leq \delta$. 
\end{assumption}

Under Assumptions \ref{cond:1}-\ref{cond:2}-\ref{cond:3}, Corollary \ref{corol:mix:hr} shows that methods in the class described by Algorithm \ref{alg:mdim:rps} can achieve a polynomial speed-up in $m$.

\begin{corollary} \label{corol:mix:hr}
 Under Assumptions \ref{cond:1}-\ref{cond:2}-\ref{cond:3}, Algorithm \ref{alg:mdim:rps} requires 
 $$n_*
 =
 O\Big(
 \frac{\kappa^{5/2} d}{m^{1/2}}
 \log^3\Big(
 \frac{2 \kappa d}{\epsilon m }
 \Big)
 \log \Big( \frac{2 \kl(\mu\mid\pi)}{\epsilon} \Big)\Big), $$
parallel rounds with $m$ zeroth-order oracles to produce a sample that is $\epsilon$-close to $\pi$ in Kullback-Leibler divergence. 
\end{corollary}

With $\epsilon$, $\kappa$ and $\kl(\mu\mid\pi)$ fixed, Corollary \ref{corol:mix:hr} shows that it is possible to sample approximately from $\pi$ with error $\epsilon$ in $O(d/m^{1/2})$ parallel rounds, up to poly-logarithmic terms. In contrast, the RWM algorithm requires $O(d)$ iterations to reach the same accuracy \citep{andrieu2024explicit}. Assuming that evaluating $\pi$ constitutes the main computational bottleneck, the random-slice approach yields a polynomial speed-up of order $m^{1/2}$. This improvement is significantly larger than the poly-logarithmic gains achieved by standard in-step parallel MCMC methods \citep{brockwell2006parallel,pozza2025fundamental}. Furthermore, compared to Algorithm \ref{alg:smala}, whose spectral gap is $O(1/(d-m))$, Algorithm \ref{alg:mdim:rps} displays a more favorable dependence on $m$, especially when $m$ is substantially smaller than $d$. This fact is also empirically validated by the results displayed in Figure \ref{fig:gain_naive}.

\begin{remark}[Optimal choice of $m$] \label{remark:m>m0}
In Algorithm \ref{alg:mdim:rps} we must select the number of random directions $m$ used at each iteration. With $m_0$ parallel processors, a natural option is to take at least $m = m_0$ (in view of Corollary \ref{corol:mix:hr}, choosing $m < m_0$ increases the total number of parallel rounds required to reach the same accuracy). However, one could also consider strategies where $m = c \times m_0$ with $c \in \N$ greater than one. In this case, following the same reasoning as above, one would need $O\left(\kappa^{3/2} (cm_0)^{1/2} \log^3(1/\delta)\right)$ sequential calls from a first-order oracle to obtain $\mathrm{KL}(\eta_{\bV,x} \mid \pi_{\bV,s^{\perp} }) \leq \delta $. In the zeroth-order framework and using $m_0$ parallel processors, the same information obtained from a single first-order oracle call can be reconstructed using $c$ parallel rounds (ignoring the finite-difference error in Algorithm \ref{alg:numdev}). This changes the total number of parallel rounds given in Corollary \ref{corol:mix:hr}, which in this case would be of order
$$
O\left(
 \frac{\kappa^{5/2} d c^{3/2} m^{1/2}_0 }{c m_0}
 \log^3\left(
 \frac{2 \kappa d}{\epsilon m }
 \right)
 \log \left( \frac{2 \kl(\mu\mid\pi)}{\epsilon} \right)\right)
 =
O\left(
 \frac{\kappa^{5/2} d c^{1/2}}{m_0^{1/2}}
 \log^3\left(
 \frac{2 \kappa d}{\epsilon m }
 \Big)
 \log \Big( \frac{2 \kl(\mu\mid\pi)}{\epsilon} \right)\right),
$$
thus being minimized for $c = 1$. Therefore, at a theoretical level, the choice $m = m_0$ is expected to be optimal, which is empirically validated by the results reported in Figure \ref{fig:eff:m:logistic} and Figure \ref{fig:eff:m:svm}. 
\end{remark}

\section{Numerics} \label{sec:5}

\subsection{Logistic regression} \label{sec:5:1}
We first investigate the performance of Algorithm \ref{alg:zohmc} on two Bayesian logistic regressions with dimension $d = 25, 200$ and sample size $n = 25, 200$, respectively. Following standard Bayesian notation, we denote the model parameters by $\beta \in \mathbb{R}^d$. As prior distribution a $\beta \sim N_d(0, \sigma^2 I_d)$ with $\sigma^2 = 25/d$ is assumed. With this choice for the likelihood and the prior, the potential function is $m$-convex, $L$-smooth and equal to
$$
U(\beta)
= \sum_{i=1}^n \Bigl[\log\bigl(1 + e^{z_i^\top \beta}\bigr)
- y_i\,z_i^\top \beta\Bigr]
+ \frac{1}{2\sigma^2}\,\beta^\top \beta \,,
$$
where $y_i \in \{0,1\}$ is the response variable and $z_i \in \mathbb{R}^d$ is the vector of explanatory variables for $i = 1,\dots,n$. In the simulation, the true parameters, $\beta_0$, are generated as $\beta_0 \sim N_d(0,(1/8) \I_d)$, each $z_i \overset{\mathrm{i.i.d.}}{\sim} N_d(0, \I_d)$, and, given $(\beta_0, z_i)$, $y_i \overset{\mathrm{i.i.d.}}{\sim} \mathrm{Bernoulli}\bigl(1/(1+\exp(-z_i^\top \beta_0))\bigr)$.

Two alternative versions of Algorithm \ref{alg:zohmc} are considered. In both, at each iteration, the matrix $\bV$ is generated by sampling uniformly and without replacement $m$ elements from the canonical basis. In the first version, the number of leapfrog iterations $L$ is set to $1$, yielding a random-slice variant of MALA (denoted RS-MALA). In the second version, $L$ is optimized to improve performance. For comparison with other zeroth-order in-step parallelization methods, we consider a multiple-try MCMC algorithm \citep{liu2000multiple} with locally balanced acceptance probability \citep{gagnon2023improving}.
As already discussed in the previous sections, \citet{pozza2025fundamental} demonstrated that the speed-up in $m$ of MTM is at most poly-logarithmic for models like the one considered. It is therefore of interest to observe whether the polynomial speed-up predicted in Section \ref{sec:4} for the random-slice approach can provide substantial gains compared to the logarithmic one of MTM in practical settings. In all three algorithms, the proposal variance is optimized adaptively \citep[Algorithm~4]{andrieu2008tutorial}. For each value of $m$ the chains are run for $5 \times 10^4$ iterations.

For a generic Markov chain $X_0, \dots, X_T$, with $X_i = (X_{i,1},\dots,X_{i,d}) \in \R^d$, we define the (empirical) expected squared jump distance (ESJD) as
\begin{equation} \label{eq:esjd}
    \text{ESJD} = \frac{1}{d(T-1)} \sum_{j = 1}^d \sum_{i = 2}^T ( X_{i+1,j} - X_{i,j} )^2.
\end{equation}
Figure \ref{fig:1} shows the ESJD for the three methods in the two models considered, normalized by the ESJD of a RWM–MCMC algorithm. For RS-HMC, the ESJD is further divided by the number of leapfrog steps $L$, so that, for any value of $m$, the three methods have approximately the same computational cost per iteration. Both RS-MALA and RS-HMC outperform MTM substantially: even with a moderate number of parallel processors, the speed-up obtained over RWM–MCMC is larger than that achieved by MTM at $m = d$, especially for the $d=200$ case. Note also that, when $m = d$, RS-MALA and RS-HMC are equivalent to the full-gradient MALA and HMC algorithms, up to the discretization error in Algorithm \ref{alg:numdev}.
In this regard, it is interesting to evaluate how the two methods approach their $m = d$ version as $m$ increases and to compare this with the corresponding behavior observed for MTM.
For MTM, improvements stabilize quickly: when $d = 200$ results at $m=75$ are similar to those at $m=200$. In contrast, RS-MALA and RS-HMC show rapid gains for $m<75$ and continue to improve noticeably as $m$ increases beyond 75. This is especially true for RS-HMC.

\begin{figure}[htbp]
  \centering
  \begin{minipage}[b]{0.48\textwidth}
    \centering
    \includegraphics[width=\linewidth]{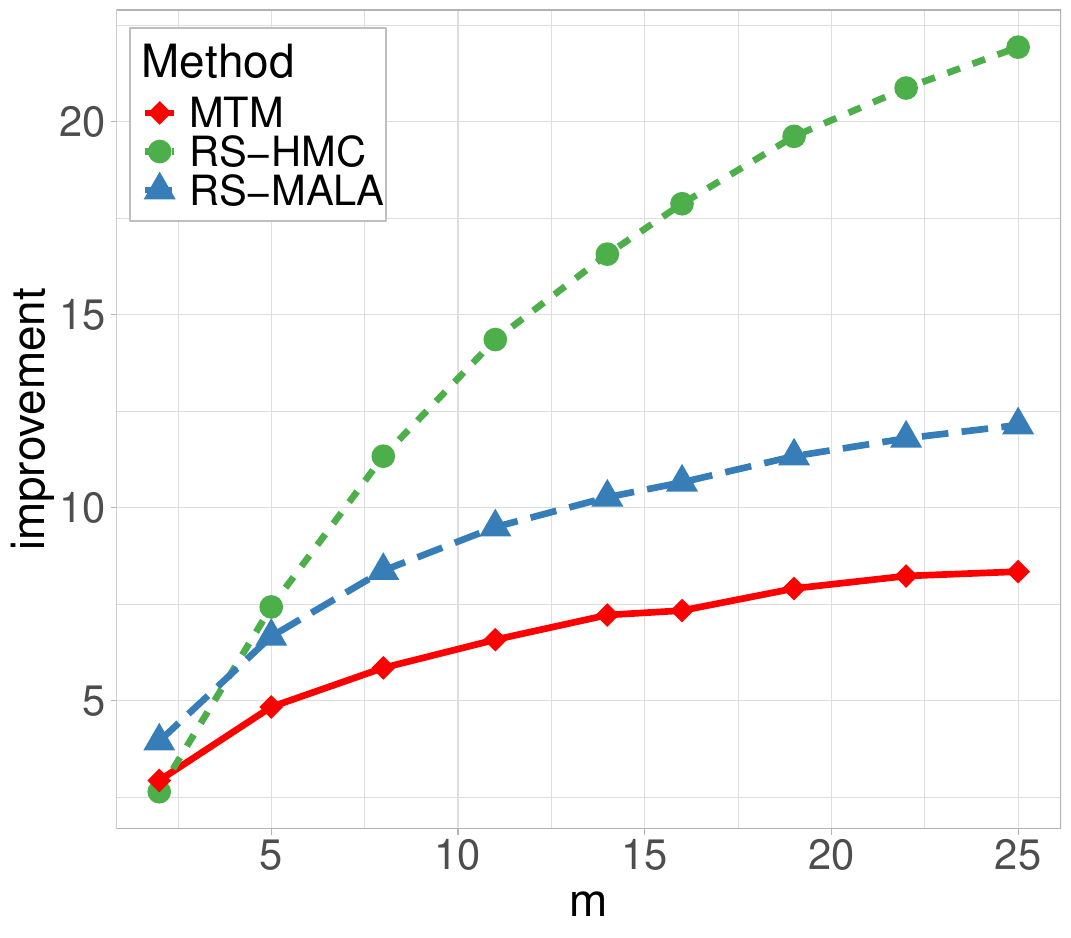}
    \vspace{-1pt}\\
    { a) d = 25}
    \label{fig:gain_d25}
  \end{minipage}
  \hfill
  \begin{minipage}[b]{0.48\textwidth}
    \centering
    \includegraphics[width=\linewidth]{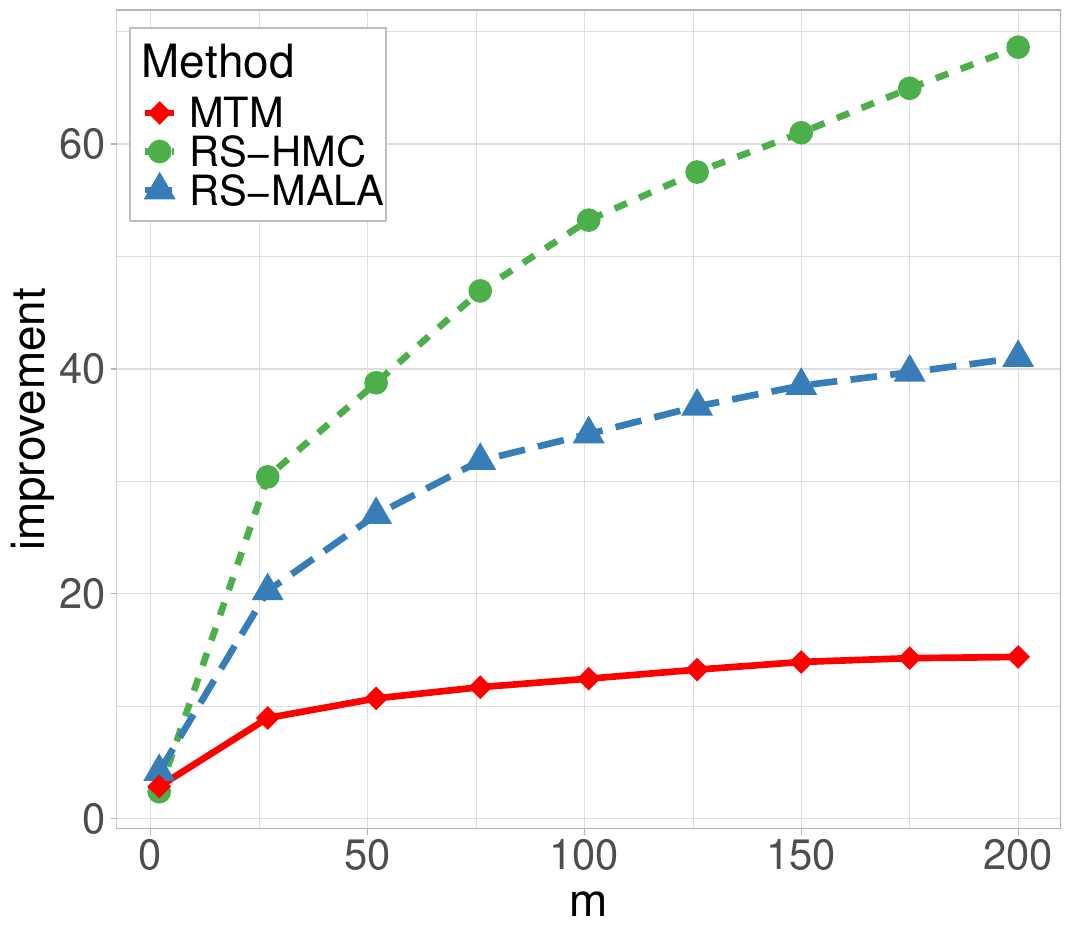}
    \vspace{-1pt}\\
    { b) d = 200}
    \label{fig:gain_d200}
  \end{minipage}
\caption{Relative improvement, over RWM, evaluated in terms of estimated average ESJD for RS-MALA, RS-HMC and MTM for the two logistic regressions considered in Section \ref{sec:5:1}. Results are standardized for the total number of $m$ parallel target evaluations at each iteration of the algorithm. }
\label{fig:1}
\end{figure}

\subsection{Stochastic volatility model}
\label{sec:5:2}
We now compare RS-MALA, RS-HMC and MTM on a stochastic volatility model for time series data \citep{kim1998stochastic}. The model consists of $n + 3$ latent parameters $\mu, \phi, \log \sigma \in \R$ and $\eta \in \R^n$ which are used to parametrize the observed volatility of a time series $\{y_i\}_{i = 1}^n$ by assuming 
\begin{equation*}
    y_i \sim N(0, \exp(h_i)), \qquad i = 1,\dots,n,
\end{equation*}
with $h_1 = \mu + \exp(\log \sigma)/\lbrace 1 - \tanh(\phi)^2 \rbrace \eta_1$ and
\begin{equation*}
    h_t = \tanh(\phi)(h_{t-1} - \mu ) + \exp( \log \sigma )\eta_t, \qquad t = 2, \dots, n.
\end{equation*}
As priors for the model parameters, we take
$$
\mu \sim N(0,10),\quad \phi \sim N(0,1),\quad \log\sigma \sim N(0,1),
\quad \eta_i \overset{\mathrm{iid}}{\sim} N(0,1)\quad i=1,\dots,n.
$$
Data are then generated by setting $n=200$, $\mu_0=1$, $\tanh(\phi_0)=0.5$, $\log\sigma_0=0$ and sampling
$\eta_{0,i} \overset{\mathrm{iid}}{\sim} N(0,1)$, thus resulting in a model with $d = 203$ parameters. For each algorithm and each value of $m$ considered, a chain of length $5\times10^4$ was obtained.
Figure \ref{fig:2} shows the performance as $m$ varies for the RS-MALA, RS-HMC and MTM algorithms. Compared with Section \ref{sec:5:1}, all three methods exhibit slightly greater improvement over RWM–MCMC. Nonetheless, RS-MALA and RS-HMC continue to outperform MTM, yielding the same qualitative conclusions as in Section \ref{sec:5:1}. Both RS-MALA and RS-HMC achieve better performance than MTM, with $m = d$, using substantially fewer target evaluations. Moreover, their performance does not stabilize as $m$ approaches $d$ but increase notably even for large values of $m$. 

\begin{figure} \label{fig:2}
    \centering
    \includegraphics[width=0.55\textwidth]{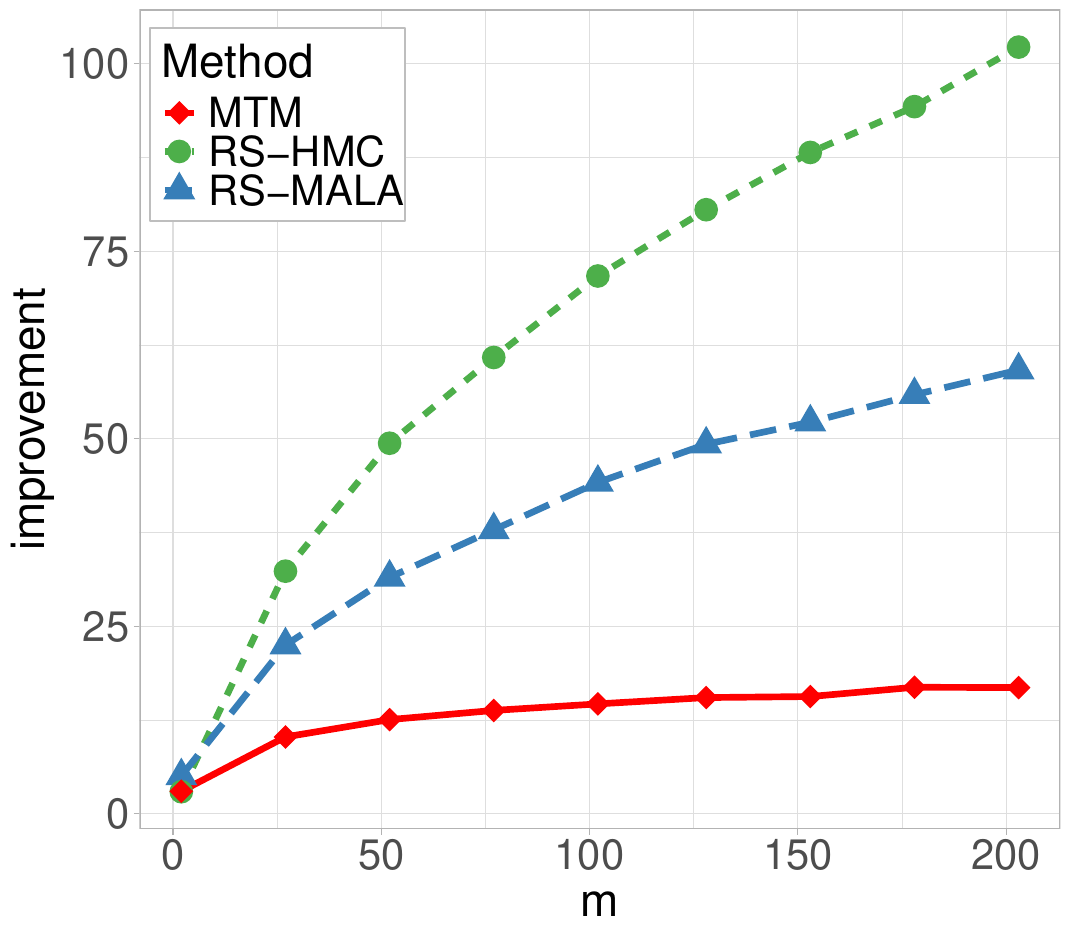}
    \caption{Relative improvement, over RWM, evaluated in terms of estimated ESJD for RS-MALA, RS-HMC and MTM for the stochastic volatility model considered in Section \ref{sec:5:2}. Results are standardized for the total number of $m$ parallel target evaluations at each iteration of the algorithm. }
\end{figure}

\subsection{Optimal choices for $m$}
The implementation of Algorithm \ref{alg:zohmc} requires specifying the number $m$ of random directions used at each iteration. With $m_0$ parallel processors, Remark \ref{remark:m>m0} suggests that setting $m=m_0$ is optimal to maximize algorithmic efficiency, at least at the theoretical level. In this section, we empirically investigate this issue by considering as measure of algorithmic efficiency the ratio between speed of convergence and cost per iteration of the resulting kernel. We empirically estimate speed of convergence using the ESJD (as defined in \eqref{eq:esjd}), thus resulting in a measure of efficiency for
Algorithm \ref{alg:zohmc} taking the form
\begin{equation} \label{eq:eff:m}
    Eff(m) 
    =
    \frac{ESJD_{P^{\textsc{rs-hmc}}_{m,L}}}{(m/m_0)L},
\end{equation}
where $ESJD_{P^{\textsc{rs-hmc}}_{m,L}}$ denotes the ESJD of the Markov chain with transition kernel $P^{\textsc{rs-hmc}}_{m,L}$. Figure \ref{fig:eff:m:logistic} and Figure \ref{fig:eff:m:svm} show the relative efficiency $Eff(m)/Eff(m_0)$ for various values of $m_0$ and $m$ in the $200$-dimensional logistic regression and in the stochastic volatility model discussed in Section \ref{sec:5:1} and Section \ref{sec:5:2}, respectively. In all the cases considered, maximal efficiency is obtained when $m = m_0$. Choosing $m > m_0$ leads to a rapid decrease of the performance of the algorithm, especially when $m$ is substantially smaller than $d$. This suggests that the simple choice $m=m_0$ is not only theoretically principled but also effective in practice.

\begin{figure}
  \centering
  \begin{minipage}{0.95\textwidth}
    \centering
    \includegraphics[width=\linewidth]{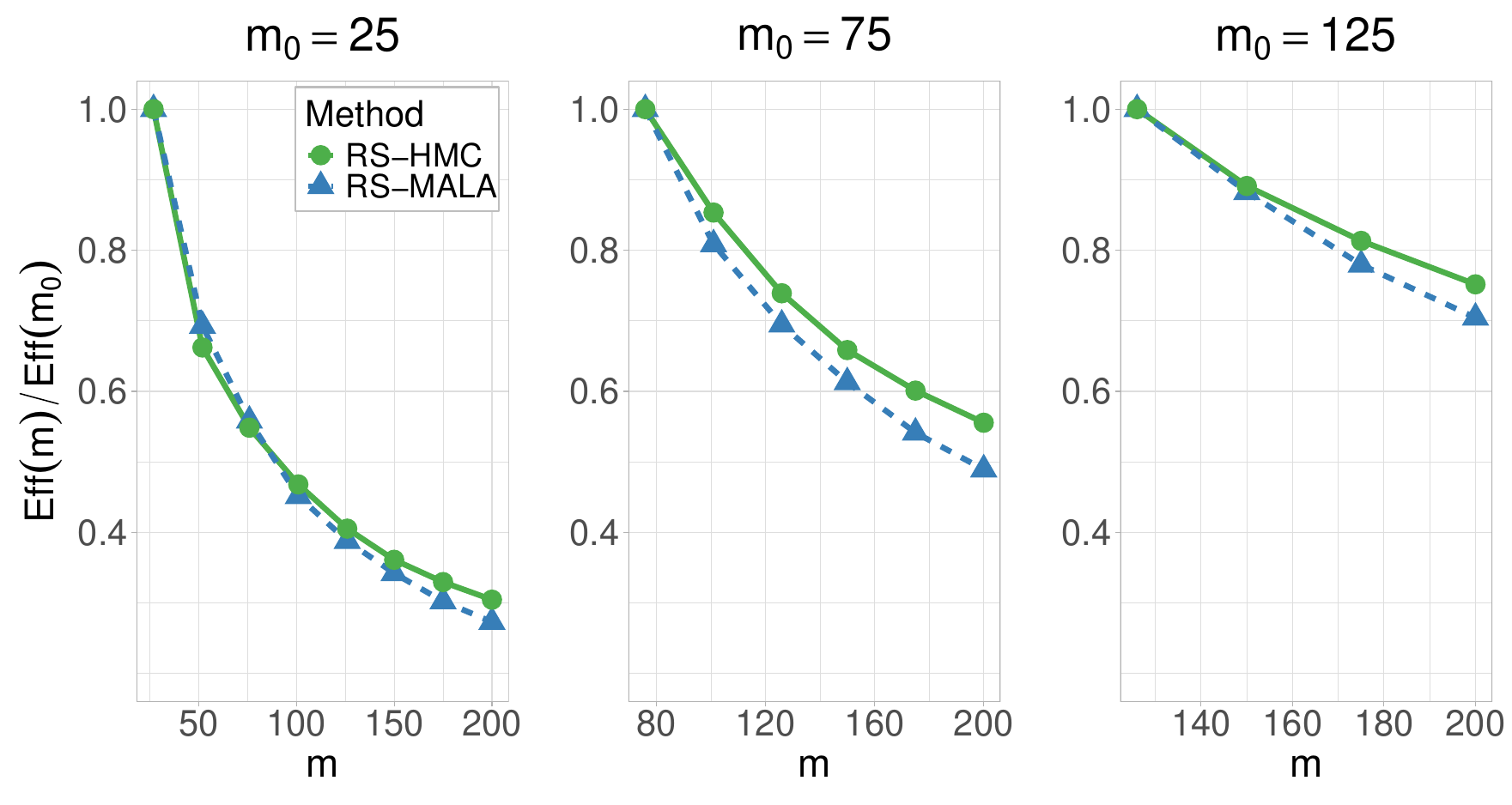}
    \caption{Relative efficiency for different values of $m_0$ in the $200$-dimensional logistic regression described in Section \ref{sec:5:1}}
    \label{fig:eff:m:logistic}
  \end{minipage}

  \vspace{\baselineskip} 
  \begin{minipage}{0.95\textwidth}
    \centering
    \includegraphics[width=\linewidth]{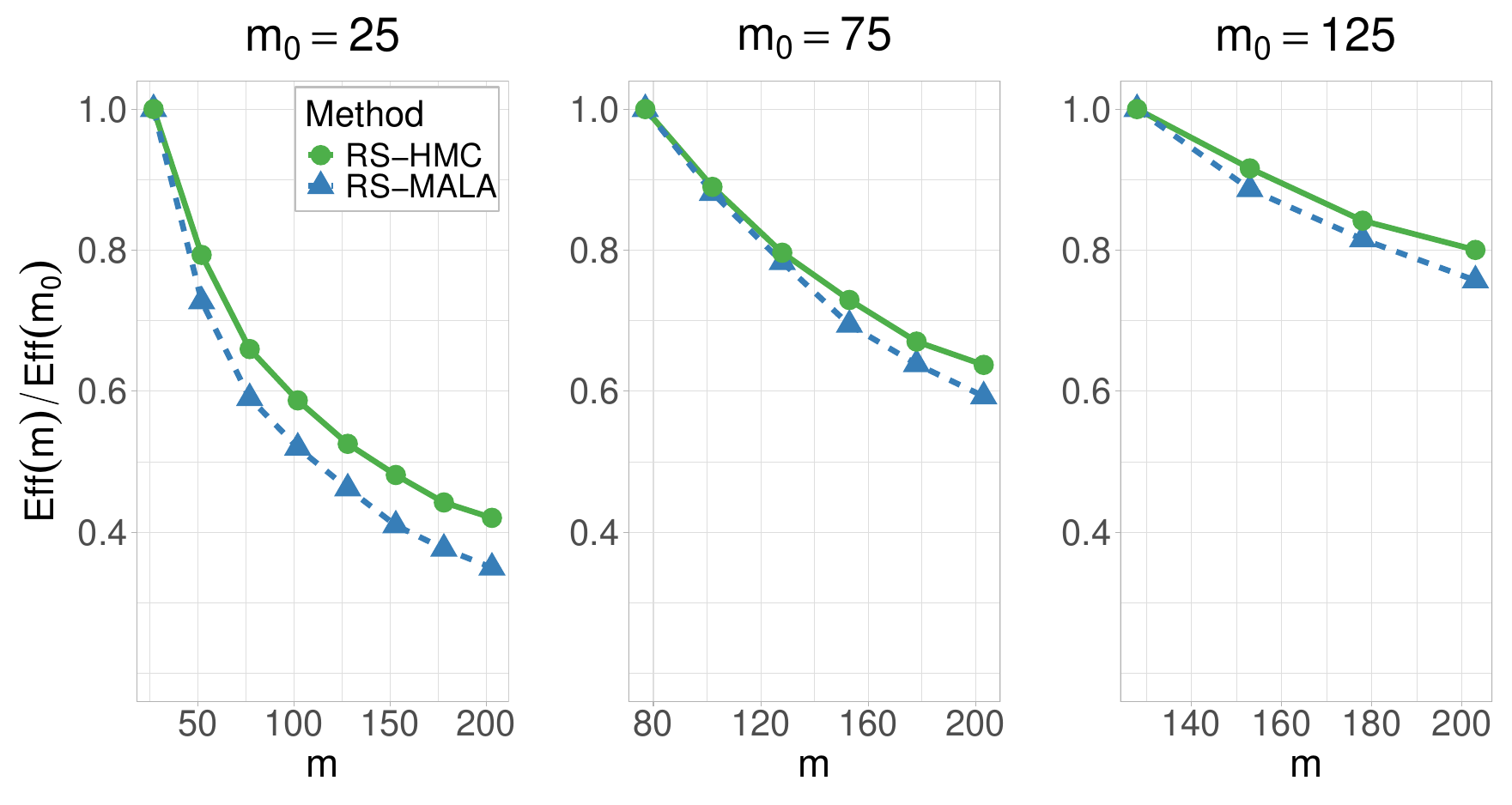}
    \caption{Relative efficiency for different values of $m_0$ in the $203$-dimensional stochastic volatility model described in Section \ref{sec:5:1}}
    \label{fig:eff:m:svm}
  \end{minipage}
\end{figure}

\paragraph*{Acknowledgements}
The authors thank Martin Chak for useful discussions and suggestions, especially regarding Theorem \ref{theo:appr:hr}.

\bibliography{bibliography}
\newpage

\appendix

\section{Proofs} \label{sec:app:a}
\subsection{Proof of Proposition \ref{prop:contr:sula}}
\begin{proof}
Let $\gamma>0$ and $x,y\in\R^d$.
We use a synchronous coupling of $X\sim \delta_{x} P^{SULA}$ and $Y\sim\delta_{y} P^{SULA}$ defined as
\begin{align*}
X&=
x -
\gamma \frac{d}{m}\bV \bV^{\top} \nabla U(x) +
\sqrt{2 \gamma} Z, 
\\
Y&=y-
\gamma \frac{d}{m}\bV \bV^{\top} \nabla U(y) +
\sqrt{2 \gamma} Z.
\end{align*}
with $Z \sim N_d(0, \I_d)$, $\bV\sim\nu$, and $Z$ and $\bV$ independent of each other.
The above coupling leads to the upper bound
\begin{align*}
W_2^2(\delta_{x} P^{SULA},\delta_{y} P^{SULA}) 
\leq &\;
\E[\|Y-X\|^2] 
\\
= &\;
\int \|
y-x
- \gamma \frac{d}{m}\bV \bV^{\top}
\left(
\nabla U(x) - \nabla U(y)
\right)
\|^2
\nu(d \bV)\,.
\end{align*}

From Assumption \ref{cond:2} and Lemma \ref{lemma:norm:cond2} we obtain
\begin{align}
&\int \|
y-x
- \gamma \frac{d}{m}\bV \bV^{\top}
\left(
\nabla U(x) - \nabla U(y)
\right)
\|^2
\nu(d \bV) 
\nonumber
\\
& = \frac{(d-m)}{d}\|
y-x\|^2
+\frac{m}{d}\|
y-x
- \gamma\frac{d}{m} (\nabla U(x)-\nabla U(y))
\|^2.
\label{eq:help:ula:1}
\end{align}
By Assumption \ref{cond:1} and \citet[][Theorem 2.1.15]{nesterov2018lectures} 
we get for $\gamma \leq m/(dL)$
\begin{equation} \label{eq:help:ula:2}
\|
y-x
- \gamma\frac{d}{m} (\nabla U(x)-\nabla U(y))
\|^2
\leq
(1- \gamma \frac{d}{m}\lambda)
\|
y-x
\|^2\,.
\end{equation}
The combination of \eqref{eq:help:ula:1} and \eqref{eq:help:ula:2} gives
\begin{align*}
 W_2^2(\delta_{x} P^{SULA},\delta_{y} P^{SULA}) 
\leq &
\frac{(d-m)}{d}\|
y-x\|^2
+\frac{m}{d}\left(1- \gamma \frac{d}{m}\lambda\right)
\|
y-x
\|^2 \\
= & 
\left(1 - \gamma \lambda  \right)
\|
y-x
\|^2,
\end{align*}
as desired.
\end{proof}

\subsection{Proof of Proposition \ref{prop:pi:inv:nsmala}}
\begin{proof}
    Let $A,B \subset \R^d$ with $A \cap B = \emptyset$. It follows from a direct application of Fubini's theorem that 
    \begin{align*}
        \int  \mathbbm{1}_{B}(x) P^{\scriptscriptstyle SMALA}_{\sigma}(x,A) \pi(dx) 
        = &
        \int_{\V(\R^d)} \Big( \int_{\R^d \times \R^d}
        \mathbbm{1}_{B}(x) 
        \mathbbm{1}_{A}(y) 
        \min
        \left(
        \pi(dx) Q_{\bV}(x,dy),
        \pi(dy) Q_{\bV}(y,dx)
        \right)
        \Big)
        \nu(d\bV)
        \\
        = &
        \int  \mathbbm{1}_{A}(y)  P^{\scriptscriptstyle SMALA}_{\sigma}(y,B)\pi(dy),
    \end{align*}
which proves the proposition.
\end{proof}
\subsection{Proof of Lemma \ref{lemma:prop:smala}}
\begin{proof}
We start from the marginal distribution of $s$. With $Q_{\bV}(x,dy)$ as in Algorithm \ref{alg:smala}, it follows from the Gaussianity of $y$ that $s$ is Gaussian with mean 
\begin{align*}
\E(s) = & \bV^{\top} \E(y) =  \bV^{\top} x - \frac{\sigma^2}{2} \bV^{\top} \hat{\nabla}_{\bV} U(x),
\end{align*}
and covariance matrix $\V(s) =  \bV^{\top} \V(y) \bV =  \sigma^2 \bV^{\top} \bV$.
This proves \eqref{eq:prop:s:smala}. Similarly, since $\bV^{\perp}$ and $\bV$ have orthogonal columns, $s^{\perp}$ is a $(d-m)$-dimensional Gaussian with mean 
\begin{align*}
\E(s^{\perp}) = & (\bV^{\perp})^{\top} \E(y)\\
= & (\bV^{\perp})^{\top} x
            - \frac{\sigma^2}{2} (\bV^{\perp})^{\top} \hat{\nabla}_{\bV} U(x)\\
         = & (\bV^{\perp})^{\top} x
         - \frac{\sigma^2 c_{\nu}}{2} (\bV^{\perp})^{\top} \bV \bV^\top \nabla U(x) \\
         = &  (\bV^{\perp})^{\top} x,
    \end{align*}
    and covariance matrix $\V(s^{\perp}) =  (\bV^{\perp})^{\top} \V(y) (\bV^{\perp})=  \sigma^2 \I_{d-m}$, which proves \eqref{eq:prop:sp:smala}. Finally, note that 
    \begin{align*}
        \E\Big( (s - \E(s) )(s^{\perp} - \E(s^{\perp}))^{\top} \Big) = &\bV^{\top}( \E(y y^{\top}) - \E(y)\E(y^{\top}) )\bV^{\perp}\\
        =& \sigma^2 \bV^{\top} \bV^{\perp} = \mathbf{0}_{m\times(d-m)},
    \end{align*}
    with the last equality following again from the fact that $\bV^{\perp}$ and $\bV$ have orthogonal columns and with $\mathbf{0}_{m\times(d-m)}$ denoting the $m\times(d-m)$ matrix of zeros. This concludes the proof of the lemma.
\end{proof}

\subsection{Proof of Theorem \ref{theo:gap:smala}}
\begin{proof}
    Recall that, for $x \sim \pi$, $s^{\perp} = (\bV^{\perp})^{\top} x$ and $s = \bV^{\top} x$. Moreover, $\pi_{\bV^{\perp}}(ds^{\perp})$ denotes the marginal distribution of $s^{\perp}$ and $\pi_{\bV,s^{\perp}}(
    ds)$ the conditional distribution of $s$ given $s^{\perp}$ having density $\pi_{\bV,s^{\perp}}(s)  \propto \exp(-U(\bV^{\perp} s^{\perp} + \bV s)) $. Using Lemma \ref{lemma:prop:smala}, $P^{\scriptscriptstyle SMALA}_{\sigma}$ can be re-written as
    \begin{align} 
       P^{\scriptscriptstyle SMALA}_{\sigma}(x,A) 
       = &
       \int_{\R^{d-m} \times \R^m } \mathbbm{1}_{A}(
       \bV^{\perp} r^{\perp} + \bV r )
       Q_{\bV^{\perp} }(s^{\perp}, dr^{\perp} )
       Q_{\bV}((s,s^{\perp}), dr ) \nonumber \\
       & \times \min \left( 1, 
       \frac{
       \pi_{\bV^{\perp}}(dr^{\perp})
       \pi_{\bV,r^{\perp}}(dr)
       Q_{\bV}((r,r^{\perp}), ds )
       }{
       \pi_{\bV^{\perp}}(ds^{\perp})
       \pi_{\bV,s^{\perp}}(ds)
       Q_{\bV}((s,s^{\perp}), dr )
       }
       \right) \nu(d \bV) \label{eq:help1:cond:smala},
    \end{align}
    for every $A \subseteq \R^{d}\backslash\{x\}$, with $Q_{\bV^{\perp} }(s^{\perp}, dr^{\perp} ) =  \phi_{d-m}\big(dr^{\perp}; s^{\perp} , \sigma^2 \I_{d-m})$
    and
    $$Q_{\bV}((s,s^{\perp}), dr ) 
    =
    \phi_{m}\big( dr; s + (\sigma^2/2) \bV^{\top} \hat{\nabla}_{\bV} U( \bV s + \bV^{\perp} s^{\perp}) , \sigma^2\I_{m} ),$$
    where the covariance of $Q_{\bV}$ is $\sigma^2\I_{m}$ since, by assumption, each column of $\bV$ is an element of the canonical basis $\mathbf{E}$, which implies $\sigma^2 \bV^{\top} \bV = \sigma^2\I_{m}$.
    Recall that $Gap_{P^{\scriptscriptstyle SMALA}_{\sigma}} \leq 2 \Phi_{P^{\scriptscriptstyle SMALA}_{\sigma}}$ \citep[][Lemma 5]{andrieu2024explicit} where 
    \begin{equation} \label{eq:def:cond}
      \Phi_{P^{\scriptscriptstyle SMALA}_{\sigma}} = 
      \inf
      \left \lbrace
        \frac{
        \int_A 
        P^{\scriptscriptstyle SMALA}_{\sigma}(x, A^c)
        \pi(dx)
        }
        {
        \pi(A)
        }
        ,
        \, 
        A \subset \R^d,
        \,
        0 < \pi(A) \leq 1/2
      \right \rbrace,
    \end{equation}
    is the conductance of $P^{\scriptscriptstyle SMALA}_{\sigma}$.
    Let
    $\sA_{1/4} = \{ A \subset \R^d \, : \, 1/4 \leq \pi(A) < 1/2  \}$.
    By \eqref{eq:def:cond}, $P^{\scriptscriptstyle SMALA}_{\sigma}$ falling in the class described by  \eqref{eq:kernel:prep:nmala} and Lemma \ref{lemma:deco:kernels} we get
    \begin{align}
        \frac{1}{2} Gap_{ P^{\scriptscriptstyle SMALA}_{\sigma}} \leq &
        \inf_{A \in \sA_{1/4}} \frac{
        \int_A 
        P^{\scriptscriptstyle SMALA}_{\sigma}(x, A^c)
        \pi(dx)
        }
        {
        \pi(A)
        } \nonumber \\
        \leq &
        \inf_{A \in \sA_{1/4}}
        4 \int
         \int
         \mathbbm{1}_{A^{\perp}}(s^{\perp}) \pi_{\bV^{\perp}}(ds^{\perp})P^{\scriptscriptstyle RW}_{\bV^{\perp}}(s^{\perp}, \R^{d-m} \backslash \{s^{\perp}\})  \nu(d\bV) 
        \label{eq:help2:cond:smala},
    \end{align}
    with $A^{\perp} = \{ s^{\perp} = (\bV^{\perp})^{\top} x \, : \, x \in A  \}$ and 
    \begin{equation*} 
        P^{\scriptscriptstyle RW}_{\bV^{\perp}}(s^{\perp}, \R^{d-m} \backslash \{s^{\perp} \}) = \int \mathbbm{1}_{\lbrace \R^{d-m} \backslash \{s^{\perp} \} \rbrace }(r^{\perp}) Q_{\bV^{\perp} }(s^{\perp}, dr^{\perp} )
        \min
        \left( 
        1,
        \frac{
        \pi_{\bV^{\perp}}(dr^{\perp})
        }{
        \pi_{\bV^{\perp}}(ds^{\perp})
        }
        \right), 
    \end{equation*}
    being a Gaussian random-walk Metropolis-Hastings kernel targeting $\pi_{\bV^{\perp}}$. Note also that from Assumption \ref{cond:1} and Lemma \ref{lemma:marginal:convex:smooth} the potential of $\pi_{\bV^{\perp}}$ is $m$-concave, $L$-smooth and twice continuously differentiable for every $\bV^{\perp}$. As a consequence, it is possible to apply Lemma \ref{lemma:ave:acc:pr:rw} to \eqref{eq:help2:cond:smala} to obtain
    \begin{align}
        Gap_{P^{\scriptscriptstyle SMALA}_{\sigma}}
        \leq
        24 \exp \Big(
        \frac{- \lambda^3 \sigma^2 (d-m) }{128 L^2 }
        \Big) + 
        32\exp\Big(
        -\frac{(d-m)}{32}
        \Big).  \label{eq:help2:cond:smala:bis}
    \end{align}
    
    Let now $g_a(x) = a^{\top}( x - \E_{\pi}(x))$ where $a \in \R^d$. From the definition of spectral gap in \eqref{eq:def:gap} it follows 
    \begin{align*}
        Gap_{P^{\scriptscriptstyle SMALA}_{\sigma}}
        \leq&
        \inf_{a \in \R^d, \, \|a\| =1} 
        \frac{\int (g_a(x)-g_a(y))^2 \pi(dx) P^{\scriptscriptstyle SMALA}_{\sigma}(x, dy)}{  2 \V_\pi(g_a)} \\
        = &
        \inf_{a \in \R^d, \, \|a\| =1} 
        \frac{\int \lbrace a^{\top}(x-y) \rbrace^2 \pi(dx) P^{\scriptscriptstyle SMALA}_{\sigma}(x, dy)}{ 2 \V_\pi(g_a)}.
\end{align*}
To upper bound the quantity in the right-hand-side of the previous display we choose to fix $a$ to be $a = (1/\sqrt{d})(1,\dots,1)$. Using \eqref{eq:help1:cond:smala} and $\bV \bV^{\top} + \bV^{\perp}(\bV^{\perp})^{\top} = \I_d$ the numerator in the right-hand-side of the previous display can be expressed as
\begin{align}
        & \int \lbrace a^{\top}(x-y) \rbrace^2 \pi(dx) P^{\scriptscriptstyle SMALA}_{\sigma}(x, dy) \nonumber \\
        & \leq
        \int \big \lbrace
         a^{\top} \bV (s - r) 
        \big \rbrace ^2
        \pi_{\bV^{\perp}}(ds^{\perp})
        \pi_{\bV,s^{\perp}}(ds)
        Q_{\bV}((s,s^{\perp}),dr)
        \nu(d\bV)\nonumber \\
        & + \int \big \lbrace
        a^{\top} \bV^{\perp} (s^{\perp} - r^{\perp}) 
        \big \rbrace ^2 \pi_{\bV^{\perp}}(ds^{\perp}) Q_{\bV^{\perp}}(s^{\perp},dr^{\perp})
        \nu(d\bV) \nonumber \\
        & + 2\int
        \big[
         a^{\top} \bV (s - r)  
        \big]
        \big[
        a^{\top} \bV^{\perp} (s^{\perp} - r^{\perp})  
        \big]
        \pi_{\bV^{\perp}}(ds^{\perp})
        \pi_{\bV,s^{\perp}}(ds) Q_{\bV^{\perp}}(s^{\perp},dr^{\perp})
        Q_{\bV}((s,s^{\perp}),dr)\nu(d\bV). \label{eq:help3:cond:smala}
    \end{align}
    We deal with the three terms in the right-hand-side of the previous display separately. The definition of $Q_{\bV^{\perp}}$, implies that, given $s^{\perp}$ and $\bV^{\perp}$, $a^{\top} \bV^{\perp}(s^{\perp} - r^{\perp}) \sim N( 0 , \sigma^2 a^{\top} \bV^{\perp}(\bV^{\perp})^{\top} a  )$ with $ a^{\top} \bV^{\perp}(\bV^{\perp})^{\top} a  \leq \|a\|^2 = 1$.
    As a result, the second element in the right-hand-side of \eqref{eq:help3:cond:smala} is upper bounded by
    \begin{equation} \label{eq:help4:cond:smala}
        \int \big \lbrace
        a^{\top} \bV^{\perp} (s^{\perp} - r^{\perp}) 
        \big \rbrace ^2 \pi_{\bV^{\perp}}(s^{\perp}) Q_{\bV^{\perp}}(s^{\perp},dr^{\perp})\nu(d\bV) \leq \sigma^2. 
    \end{equation}
    Second, from Lemma \ref{lemma:prop:smala}, given $s$, $s^{\perp}$ and $\bV$, the random  variables
    $\big[
         a^{\top} \bV (s - r)  
    \big]
    $
    and
    $
    \big[
        a^{\top} \bV^{\perp} (s^{\perp} - r^{\perp})  
    \big]$
    are independent which implies that the third element in the right-hand-side of \eqref{eq:help3:cond:smala} is zero.
    
   The first element the right-hand-side of \eqref{eq:help3:cond:smala} can be further decomposed as
   {\small
   \begin{align}
   & \int \big \lbrace 
    a^{\top} \bV (s - r)  
    \big \rbrace ^2
    \pi_{\bV^{\perp}}(ds^{\perp})
    \pi_{\bV,s^{\perp}}(ds)
    Q_{\bV}((s,s^{\perp}),dr) 
    \nu(d\bV) \nonumber \\
    =& \int \big \lbrace
    a^{\top}
    \bV 
    \big(s - \frac{\sigma^2}{2} \frac{d}{m} \bV^{\top} \nabla U(\bV 
       s  + \bV^{\perp} s^{\perp})  - r \big)
    \big \rbrace ^2
    \pi_{\bV^{\perp}}(ds^{\perp})
    \pi_{\bV,s^{\perp}}(ds)
    Q_{\bV}((s,s^{\perp}),dr)
    \nu(d\bV) \nonumber \\
    + & \int \big \lbrace
    \frac{\sigma^2}{2}
   \frac{d}{m}
    a^{\top}
    \bV
   \bV^{\top} \nabla U(\bV 
       s  + \bV^{\perp} s^{\perp})
    \big \rbrace ^2
    \pi_{\bV^{\perp}}(ds^{\perp})
    \pi_{\bV,s^{\perp}}(ds)
    Q_{\bV}((s,s^{\perp}),dr)
    \nu(d\bV) \nonumber \\
    + & 2 \int 
    \big \lbrace
    a^{\top}
    \bV
    \big(
    \frac{\sigma^2}{2} \frac{d}{m} \bV^{\top} \nabla U(\bV s  + \bV^{\perp} s^{\perp})
    \big)  
    \big \rbrace 
    \big \lbrace
     a^{\top}
     \bV
    \big(s - 
    \frac{\sigma^2}{2} \frac{d}{m} \bV^{\top} 
    \nabla U(\bV s  + \bV^{\perp} s^{\perp})
    - r \big) 
    \big \rbrace 
    \nonumber
    \\
    & \qquad \times
    \pi_{\bV^{\perp}}(ds^{\perp})
    \pi_{\bV,s^{\perp}}(ds)
    Q_{\bV}((s,s^{\perp}),dr)
    \nu(d\bV).
    \label{eq:help4:cond:smala:bis:1}
    \end{align}  }
   
    Conditioned on $(s, s^{\perp} )$ and $\bV$, Lemma \ref{lemma:prop:smala} implies 
    \begin{equation} \label{eq:help4:1:cond:smala}
       a^{\top} \bV 
       \Big(s - 
       \frac{\sigma^2}{2}
       \frac{d}{m}
       \bV^{\top} 
       \nabla U(\bV 
       s  + \bV^{\perp} s^{\perp})
       -
       r\Big)
       \sim
        N( 0 , \sigma^2 a^{\top} \bV \bV^{\top} a ).
    \end{equation}
    Moreover, from Assumption \ref{cond:2} and $\| a \| = 1$
    \begin{equation} \label{eq:help4:2:cond:smala}
        \int \big(
        a^{\top} \bV \bV^{\top} a
        \big)
        \nu(d \bV) = 
        \frac{m}{d} \|a\|^2
        =
        \frac{m}{d}.
    \end{equation}
   From \eqref{eq:help4:1:cond:smala} and \eqref{eq:help4:2:cond:smala} it follows
    \begin{align} 
    &\int \big \lbrace
    a^{\top}
    \bV 
    \big(s - \frac{\sigma^2}{2} \frac{d}{m} \bV^{\top} U(\bV 
       s  + \bV^{\perp} s^{\perp})  - r \big)
    \big \rbrace ^2
    \pi_{\bV^{\perp}}(ds^{\perp})
    \pi_{\bV,s^{\perp}}(ds)
    Q_{\bV}((s,s^{\perp}),dr)
    \nu(d\bV) \nonumber \\
    = & \,
    \sigma^2  \int \big( a^{\top}\bV \bV^{\top}a \big) \nu(d\bV) 
    \leq \,
    \sigma^2  \frac{m}{d}. \label{eq:help4:cond:smala:bis:2}
    \end{align}
Similarly,
    \begin{align}
    & \int 
    \big \lbrace
    a^{\top}
    \bV
    \big(
    \frac{\sigma^2}{2} \frac{d}{m} \bV^{\top} \nabla U(\bV s  + \bV^{\perp} s^{\perp})
    \big)  
    \big \rbrace 
    \big \lbrace
     a^{\top}
     \bV
    \big(s - 
    \frac{\sigma^2}{2} \frac{d}{m} \bV^{\top} 
    \nabla
    U(\bV s  + \bV^{\perp} s^{\perp})
    - r \big) 
    \big \rbrace 
    \nonumber
    \\
    &
    \qquad
    \times 
    \pi_{\bV^{\perp}}(ds^{\perp})
    \pi_{\bV,s^{\perp}}(ds)
    Q_{\bV}((s,s^{\perp}),dr)
    \nu(d\bV)
    =
    0.
   \label{eq:help4:cond:smala:bis:3}
    \end{align}
    Finally, using that, by assumption, every column of $\bV$ belongs to the canonical basis of $\R^d$ combined with $ x = \bV s + \bV^{\perp} s^{\perp} $, $a = (1/\sqrt{d})(1,\dots,1)$, Cauchy–Schwarz inequality and Assumption \ref{cond:2} give
    \begin{align}
    & \int \big \lbrace
   \frac{\sigma^2}{2}
   \frac{d}{m}
    a^{\top}
    \bV
   \bV^{\top} \nabla U(\bV 
       s  + \bV^{\perp} s^{\perp})
    \big \rbrace ^2
    \pi_{\bV^{\perp}}(ds^{\perp})
    \pi_{\bV,s^{\perp}}(ds)
    Q_{\bV}((s,s^{\perp}),dr)
    \nu(d\bV) \nonumber \\
    \leq & \,
    \frac{\sigma^4}{4} 
    \frac{d^2}{m^2}
    \int
    (a^{\top} \bV \bV^{\top} a)
    ( \nabla U(x)^{\top} \bV \bV^{\top} \nabla U(x) ) 
    \nu(\bV)
    \pi(dx)
    \nonumber
    \\
    =
    &
    \, \frac{\sigma^4}{4}
    \int
    \| \nabla U(x) \|^2
    \pi(dx) 
    \, \leq \,
    \sigma^4
    \frac{L^2}{\lambda}
     \frac{d}{4}, \label{eq:help4:cond:smala:bis:4} 
    \end{align}
    with the last inequality following from the fact that, from Assumption \ref{cond:1} and Lemma \ref{lemma:norm:ucon}, $\| \nabla U(x) \|^2 \leq L^2 \|x - x_* \|^2$ and $\|x - x_* \|^2 \overset{s}{\preccurlyeq} \lambda^{-1} \chi^2_d $ where $\chi_d^2$ is a chi-squared random variable with $d$ degrees of freedom, $x_* \in \R^d$ is the minimizer of $U(x)$ and $\overset{s}{\preccurlyeq}$ denotes the stochastic ordering. Thus, combining \eqref{eq:help4:cond:smala:bis:2}-\eqref{eq:help4:cond:smala:bis:3}-\eqref{eq:help4:cond:smala:bis:4} we obtain
    \begin{equation} \label{eq:help5:cond:smala}
    \int \big \lbrace 
    a^{\top} \bV (s - r)  
    \big \rbrace ^2
    \pi_{\bV^{\perp}}(ds^{\perp})
    \pi_{\bV,s^{\perp}}(ds)
    Q_{\bV}((s,s^{\perp}),dr) 
    \nu(d\bV)
    \leq
    \sigma^2 \frac{m}{d} +
    \sigma^4
    \frac{L^2}{\lambda}
    \frac{d}{4}.
    \end{equation}

    Finally, $\mathrm{Var}({a^{\top}}  X) \geq \|a\|^2/L= 1/L$ from Assumption \ref{cond:1} and Equation 10.25 of \citet{saumard2014log}. Combining this with \eqref{eq:help2:cond:smala:bis}, \eqref{eq:help4:cond:smala} and \eqref{eq:help5:cond:smala} gives 
    \begin{equation} \label{eq:help6:cond:smala}
    Gap_{P^{\scriptscriptstyle SMALA}_{\sigma}}
    \leq
    \min\Big(
    f_1(\sigma)
    ,
    f_2(\sigma)
    \Big),
    \end{equation}
    where
    \begin{align*}
    f_1(\sigma) 
    = & \,
    24 \exp \Big( 
    \frac{- \lambda^3 \sigma^2 (d-m) }{128 L^2 }
    \Big) + 
    32\exp\Big(
    -\frac{(d-m)}{32}
    \Big), \\
    f_2(\sigma) 
    = & \,
    \sigma^2 \frac{L}{2} \frac{m}{d} +
    \frac{\sigma^4}{8} \frac{d L^3}{ \lambda}.
    \end{align*}
    Note that $f_1(\sigma)$ is positive, continuous and monotonically decreasing in $\sigma$ while $f_2(\sigma)$ is positive, continuous and monotonically increasing in $\sigma$. Moreover, $\lim_{\sigma \to 0^+} f_{1}(\sigma)/ f_{2}(\sigma) = + \infty$ and $\lim_{\sigma \to 0} f_{1}(\sigma)/ f_{2}(\sigma) = 0$. 
    This and \eqref{eq:help6:cond:smala} imply
   \begin{equation} \label{eq:help7:cond:smala}
   \sup_{\sigma > 0 } Gap_{P^{\scriptscriptstyle SMALA}_{\sigma}}
    \leq 
    \sup_{\sigma > 0 }
    \min\Big(
    f_1(\sigma)
    ,
    f_2(\sigma)
    \Big)
    \leq
    \max \Big(
    f_1(\sigma)
    ,
    f_2(\sigma)
    \Big),
    \end{equation}
   for every $\sigma > 0 $.
   To prove \eqref{eq:theo:gap:smala:result2}, let $\sigma^2_* = 128 L^2 \log(d-m) /\lbrace \lambda^3 (d-m) \rbrace$. The fact that $m < c_0 d$ implies $m/(d-m) < c_1$ (for $c_1 = c_0/(1-c_0)$). Combining this with \eqref{eq:help7:cond:smala} give
    \begin{align*}
        \sup_{\sigma}
        Gap_{P^{\scriptscriptstyle SMALA}_{\sigma}} 
        \leq &
        \max \Big(
        f_1(\sigma^*)
        ,
        f_2(\sigma^*)
        \Big) \\
        = &
        \max\Big(
        \frac{24}{d-m} + 
        32\exp\Big(
        -\frac{(d-m)}{32}\Big),
         \frac{64 c_0 L^3 \log(d-m)}{(d-m)} + 
          \frac{2048 L^{7} \log(d-m)^2}{ \lambda^{7} (d-m)^2} d 
        \Big) \\
        = & O\Big( \frac{\log(d-m)^2}{(d-m)} \Big),
    \end{align*}
    which concludes the proof of the theorem.
    \end{proof}

\subsection{Proof of Theorem \ref{theo:appr:hr}}
\begin{proof}
    Given a generic probability measure $\mu \in \sP(\R^d)$ and $\bV \in \V_m(\R^d)$, we denote with $\mu_{\bV^{\perp}}$ the marginal distribution of $s^{\perp} = \bV^{\perp} x $ for $x \sim \mu$ and with $\mu_{\bV,s^{\perp}}$ the conditional distribution of $s = \bV^{\top} x$ given $s^{\perp}$. We start by considering the case $t = 1$. Exploiting the structure of $P(x,\cdot)$ given in \eqref{eq:kernel:bbhar}, the convexity and the chain-rule properties of the Kullback-Leibler divergence together with the fact that \smash{$(\mu P_{\bV})_{\bV^{\perp}} 
    =
    \mu_{\bV^{\perp}} $} and \smash{$(\mu P)_{\bV, s^{\perp}} 
    =
    \eta_{\bV, x}$}
    (with \smash{$\eta_{\bV,x}$} defined as in \eqref{eq:cond:theorem:appr:hr} )
    we get 
    \begin{align*}
        \kl(\mu P \mid \pi )
        \leq
        &
        \int \kl(\mu P_{\bV} \mid \pi ) \nu(d\bV)
        \\
        = &  \int \kl( \mu_{\bV^{\perp}}
        \mid 
        \pi_{\bV^{\perp}} ) \nu(d\bV) 
        + \int \E_{s^{\perp}
        \sim
        \mu_{\bV^{\perp}}}
        \left[
        \kl(
        \eta_{\bV, x }
        \mid
        \pi_{\bV, s^{\perp}}
        )
        \right]
        \nu(d\bV).
    \end{align*}
    Under Assumptions \ref{cond:1}-\ref{cond:2}, it follows from the results given in the proof of Corollary 8.2 of \citet{ascolani2024entropy} that
    \begin{equation*}
        \int \kl(  \mu_{\bV, s^{\perp}}
        \mid 
         \pi_{\bV, s^{\perp}} ) \nu(d\bV) 
        \leq
        \left(
        1 - \frac{1}{\kappa} \frac{m}{d} 
        \right)
        \kl( \mu \mid \pi ).
    \end{equation*}
    From \eqref{eq:cond:theorem:appr:hr}, $\kl(
        \eta_{\bV, x }
        \mid
          \pi_{\bV, s^{\perp} }
         ) < \delta$, which implies
    \begin{equation*}
        \kl(\mu P \mid \pi ) \leq \left( 1 - \frac{1}{\kappa} \frac{m}{d} \right)\kl( \mu \mid \pi ) + \delta.
    \end{equation*}
   Iterating the same reasoning for $t>1$ (replacing $\mu$ with $\mu P^{t-1}$) gives
    \begin{align}
        \kl(\mu P^t \mid \pi ) \leq & \Big( 1 - \frac{1}{\kappa} \frac{m}{d} \Big)^t \kl(\mu \mid \pi ) + \delta \sum_{i = 1}^{t-1} \Big( 1 - \frac{1}{\kappa} \frac{m}{d} \Big)^i \nonumber \\
        \leq & \Big( 1 - \frac{1}{\kappa} \frac{m}{d} \Big)^t \kl(\mu \mid \pi ) +\frac{\delta \kappa d }{m}, \label{eq:help0:thm:appr:hr}
    \end{align}
    with the last inequality following from $1 - m/(\kappa d) < 1$. This proves \eqref{eq:result1:theorem:appr:hr}.  
    
    To conclude, observe that from $(1-x) \leq \exp(-x)$. As a result, the first term in the right-hand-side of \eqref{eq:result1:theorem:appr:hr} is upper bounded by $\exp(-t m/(\kappa d))$. Thus,
    \begin{equation} \label{eq:help1:thm:appr:hr}
       \kl(\mu P^t \mid \pi ) \leq \frac{\epsilon}{2}, \quad  \mathrm{if} \quad t \geq (\kappa d)/m \log(2 \kl(\mu\mid\pi)/\epsilon),
    \end{equation}
    for any $\epsilon>0$.
    The theorem is proven by combining \eqref{eq:help0:thm:appr:hr} with \eqref{eq:help1:thm:appr:hr} and by choosing $\delta \leq \epsilon m/(2 \kappa d) $
\end{proof}

\subsection{Proof of Corollary \ref{corol:mix:hr}}
\begin{proof}
    The result follows directly from Theorem \ref{theo:appr:hr} and Assumption \eqref{cond:3}.
\end{proof}

\subsection{Proof of Proposition \ref{prop:rshmc}}
\begin{proof}
  To prove the statement we show that, given $\bV \in \V_m(\R^{d})$ and $x \in \R^d$, every iteration of Algorithm \ref{alg:rshmc} is $\pi_{\bV,s^{\perp}}$-invariant. With this regard, first note that if we define $T_L$ to be the map for which $(s_L,k'_L) = T_L(s_0,k)$ (where $(s_0,k)$ and $(s_L,k'_L)$ are as described in Algorithm \ref{alg:rshmc}) then $T_L^{-1} = T_L$. Moreover, $T_L$ is the result of multiple compositions of the functions $F_{1}(s,k) = (s + k, k)$ and $F_{2,\gamma_*}(s,k) = ( s \,, \, k + \gamma_* \hat{\nabla}_{\bV} U(x + \bV s) )$, $\gamma_* > 0$. The determinants of the Jacobian matrices of $F_{1}$ and $F_{2,\gamma_*}$, satisfy $|J_{F_{1}}(s,k)| = 1$ and $|J_{F_{2,\gamma_*}}(s,k)| = 1$ for every $(s,k) \in \R^{2m}$ which implies (in view of the chain rule for the Jacobian matrix) $|J_{T_L(s_0,k)}| = 1$ for every $(s_0,k) \in \R^{2m}$. As a consequence, the transition rule of Algorithm \ref{alg:rshmc} belongs to the family of deterministic proposals discussed in \citet{tierney1998note}. Thus, choosing 
  $$ 
  \min\left(
  1,
  \exp\Big( U(x + \bV s_0) - U(x + \bV s_L) + 0.5(\|k\|^2 - \|k'_{L}\|^2) \Big)\right),
$$
  as acceptance probability makes the transition kernel reversible with respect to the joint distribution $\pi_{\bV,s^{\perp}} \otimes N_m(0,\I_d)$ and, therefore, $\pi_{\bV,s^{\perp}}$-reversible. This concludes the proof of the proposition.
\end{proof}

\section{Technical results}
\subsection{Efficiency of decomposable kernels}
Let $(X_1, X_2) = X \sim \pi$, $\pi \in \mathcal{P}(\R^d)$. We denote with $\pi_1 \in \mathcal{P}(\R^{d-p})$ the marginal distribution of $X_1$ and with $\pi_{2,x_1}$ the conditional distribution of $X_2$ given $X_1 = x_1$.  In the following, we consider kernels of the form 
\begin{align}
    P(x,A) = & \int_{\R^{m\times (d-m)}} \mathbbm{1}_A( (y_1,y_2) )
    Q_{1}(x_1, dy_1)
    Q_{2,y_1}(x, dy_2) \nonumber \\
    & \times
    \min\Big(1, 
    \frac{
    \pi_{1}(dy_1)
    \pi_{2,y_1}(dy_2)
    Q_{1}(y_1, dx_1)
    Q_{2,x_1}(y,dx_2)
    }{
    \pi_{1}(dx_1)
    \pi_{2,x_1}(dx_2)
     Q_{1}(x_1, dy_1)
    Q_{2,y_1}(x, dy_2)
    }  
 \Big), \label{eq:kernel:prep:nmala}
\end{align}
where $Q_{1}\, : \, \R^{d-m} \to \sP\big(\R^{d-m} \big)  $ and $Q_{2,y_1}\, : \, \R^m \times \R^{d-m} \to \sP\big(\R^{m} \big) $ are absolutely continuous with respect to the Lebesgue measure. Note that $Q_{1}$ depends only on the previous value of $x_1$ while $Q_{2,y_1}$ can depend on $x_1$, $x_2$ and $y_1$.

\begin{lemma} \label{lemma:deco:kernels}
    Let $\pi\in\sP(\R^d)$ be absolutely continuous with respect to Lebesgue and $P$ as in \eqref{eq:kernel:prep:nmala}. Then, for every $x_1 \in \R^{d-m}$, 
    \begin{equation} \label{eq:result1:lemma:deco:kernels}
        \int_{\R^{m} } \pi_{2,x_1}(dx_2) P( (x_1,x_2), \R^{d} \backslash \{x\})
        \leq
        P_1(x_1, \R^{d-m}\backslash \{x_1\}),
    \end{equation}
    where 
    \begin{align*}
        P_1(x_1, \R^{d-m} \backslash \{x_1\} ) =
        \int 
        \mathbbm{1}_{\lbrace \R^{d-m} \backslash \{x_1\} \rbrace }(x_1)
        Q_{1}(x_1,dy_1)
        \min\Big(1, 
        \frac{
        \pi_{1}(dy_1)
        Q_{1}(y_1, dx_1 )
        }{
        \pi_{1}(dx_1)
        Q_{1}(x_1 , dy_1)
        }  
        \Big).
\end{align*}
\end{lemma}

\begin{proof}
    From the definition of $P$ in \eqref{eq:kernel:prep:nmala} and $\pi$ being absolutely continuous with respect to the Lebesgue measure, Fubini theorem and an application of Jensen inequality to $\min(1,t)$ give
    \begin{align*}
    \int_{ \R^{m} }
    \pi_{2,x_1}(d x_2) P(x, \R^{d} \backslash \{x\}) 
    = & \int_{ \R^{m} }
    \int_{ \R^{d} \backslash \{x\}}
    \pi_{2,x_1}(dx_2)  
    Q_{1}(x_1, dy_1)
    Q_{2,y_1}(x, dy_2) \nonumber \\
    & \times
    \min\Big(1, 
    \frac{
    \pi_{1}(dy_1)
    \pi_{2,y_1}(dy_2)
    Q_{1}(y_1, dx_1)
    Q_{2,x_1}(y,dx_2)
    }{
    \pi_{1}(dx_1)
    \pi_{2,x_1}(dx_2)
     Q_{1}(x_1, dy_1)
    Q_{2,y_1}(x, dy_2)
    } 
    \Big) \\
    \leq &
    \int_{\R^{d-m} \backslash \{x_1\}} 
    Q_{1}(x_1, dy_1)
    \min\Big(1,
    \frac{
    \pi_{1}(dy_1)
    Q_{1}(y_1, dx_1)
    }{
    \pi_{1}(dx_1)
     Q_{1}(x_1, dy_1)
    } 
    \\
    & \times  
    \int_{\R^{m} \times \R^{m} }
    \pi_{2,y_1}(dy_2)
     Q_{2,x_1}(y,dx_2)
    \Big) \\
    \leq & 
    \int_{\R^{d-m} \backslash \{x_1\} } 
    Q_{1}(x_1, dy_1)
    \min\Big(1,
    \frac{
    \pi_{1}(dy_1)
    Q_{1}(y_1, dx_1)
    }{
    \pi_{1}(dx_1)
     Q_{1}(x_1, dy_1)
    } 
    \Big),
    \end{align*}
    with the last inequality following from $\int_{\R^{m} \times \R^{m}}
    \pi_{2,y_1}(dy_2)
     Q_{2,x_1}(y,dx_2) = 1$. 
     This concludes the proof.
\end{proof}

\subsection{Additional technical results}
\begin{lemma}[Average acceptance RWM]
\label{lemma:ave:acc:pr:rw}
    Let $\pi$ satisfying Assumption \ref{cond:1} and $P^{\textsc{rw}}$ being a random-walk Metropolis–Hastings transition kernel
    \begin{equation} \label{eq:lemma:average:acc:rw:kernel}
        P^{\textsc{rw}}(x, A )
        =
        \int_{\R^d} \mathbbm{1}_{A}(x) \phi_d(y; x, \sigma^2 \I_d)
        \min\Big(
        1,
        \exp( U(x) - U(y) )
        \Big)
        dy, 
    \end{equation}
    for every $A \subset \R^d$. Then 
      \begin{equation} 
      \label{eq:lemma:average:acc:rw:result1}
         P^{\textsc{rw}}(x, \R^d \backslash \{x\} )
         \leq
         \exp\left( -\frac{\lambda \sigma^2 d}{4} \right)
         +2 \left \lbrace
         \exp\left(
         - \frac{(\lambda \sigma d)^2}{64 \|\nabla U(x)\|^2}
         \right) + 
         \exp\left( - \frac{d}{32}  \right) \right \rbrace,
    \end{equation}
    and
    \begin{equation} \label{eq:lemma:average:acc:rw:result2}
        \int_{\R^d}
        \pi(dx)
         P^{\textsc{rw}}(x, \R^d \backslash \{x\} ) 
         \leq
        3 \exp \Big(
        \frac{- \lambda^3 \sigma^2 d }{128 L^2 }
        \Big) + 
        4\exp\Big(
        -\frac{d}{32}
        \Big).
    \end{equation}
\end{lemma}

\begin{proof}

We start by decomposing $y$ in \eqref{eq:lemma:average:acc:rw:kernel} as $y = x + \sigma z$ where $z \sim N_d(0, \I_d)$. From $U$ being $m$-convex it follows that 
\begin{align*}
       U(x) - U(x + \sigma z)
       \leq &
       - \sigma \nabla U(x)^{\top} z
       - \frac{\lambda \sigma^2}{2}\|z\|^2  \\  
       = &\, -\sigma  \nabla U(x)^{\top} z
       - \frac{\lambda \sigma^2 d}{2}
       - \frac{\lambda \sigma^2}{2}
       \sum_{i = 1}^d (|z_i|^2 -1).
\end{align*}
This, together with an application of the union bound, implies that
\begin{align}
    \Pr\left(
    \exp\big( U(x) - U(x + \sigma z) \big)
    \geq
    \exp\Big(
    - \frac{\lambda \sigma^2 d}{4} 
    \Big)
    \right)
    \leq &
    \Pr\left(
    \sigma  \nabla U(x)^{\top} z 
    + \frac{\lambda \sigma^2 d}{2}
    + \frac{\lambda \sigma^2}{2}
    \sum_{i = 1}^d (|z_i|^2 -1)
    <
    \frac{\lambda \sigma^2 d}{4}
    \right) \nonumber \\
    =& \Pr\left(
    \sigma  \nabla U(x)^{\top} z
    +\frac{\lambda \sigma^2}{2} 
    \sum_{i = 1}^d (|z_i|^2 - 1 )
    <
    -\frac{\lambda \sigma^2 d}{4}
    \right) \nonumber \\
    \leq &
    \Pr\left(
    \nabla U(x)^{\top} z
    <
    -\frac{\lambda \sigma d}{8}
    \right)
    + \Pr\left(
    \frac{1}{d} \sum_{i = 1}^d (|z_i|^2 - 1 )
    < -\frac{1}{4}
    \right) \nonumber \\
    \leq & 2 \left \lbrace
    \exp \left(
    - \frac{( \sigma \lambda d)^2 }
    {64 \|\nabla U(x)\|^2}
    \right)
    + \exp\left(
    -\frac{d}{32}
    \right)
    \right\rbrace \label{eq:lemma:average:acc:rw:help0},
\end{align}
with the last line following from Bernstein's inequality.
Leveraging \eqref{eq:lemma:average:acc:rw:help0} and using that, for every $t>0$, $\E(\min(1,x)) \leq t + \Pr(x>t)$ we obtain
 \begin{align*}
       P^{\textsc{rw}}\big(x, \R^d \backslash \{x\} \big)
       \leq
       \exp\Big(
       -\frac{\lambda \sigma^2 d}{4}
       \Big)
       + 2
       \Big \lbrace
       \exp \Big(
       \frac{-(\sigma \lambda d)^2 }{64 \|\nabla U(x)\|^2} \Big)
       + 
       \exp\Big(
       -\frac{d}{32}
       \Big)
       \Big
       \rbrace,
\end{align*}  
which proves \eqref{eq:lemma:average:acc:rw:result1}. As regards \eqref{eq:lemma:average:acc:rw:result2}, let $x_*$ be the unique minimizer of $U(x)$ and recall that for $L$-smooth potential functions $\|\nabla U(x)\|^2 \leq L^2 \| x - x_* \|^2$. Moreover, from Assumption \ref{cond:1} and Lemma \ref{lemma:norm:ucon} it follows that if $X \sim \exp(-U)$ then $\| X - x_* \|^2 \overset{s}{\preccurlyeq} \lambda^{-1} \chi^2_d$. By combining these observations with \eqref{eq:lemma:average:acc:rw:help0}, the fact that $\exp(-1/k)$ is monotone increasing in $k$ for $k>0$, and Lemma \ref{lemma:chiq}  we obtain
\begin{align*}
  \int_{\R^d}
  \pi(dx)
  P^{\textsc{rw}}(x, \R^d \backslash \{x\} )
  \leq &
  \exp\Big(
   -\frac{\lambda \sigma^2 d}{4}
   \Big)
   + 
   2\exp\Big(
   -\frac{d}{32}
   \Big)
   + 2
   \int 
   \pi(dx)
   \exp \Big(
   \frac{-(\sigma \lambda d)^2 }{64 \|\nabla U(x)\|^2} \Big)
   \\
  \leq &
  \exp\Big(
   -\frac{\lambda \sigma^2 d}{4}
   \Big)
   + 
   2\exp\Big(
   -\frac{d}{32}
   \Big)
   + 2
   \int 
   \pi(dx)
   \exp \Big(
   \frac{- \lambda^3 (\sigma d)^2 }{64 L^2 \|x - x_* \|^2 } \Big)
    \\
    \leq &
  \exp\Big(
   -\frac{\lambda \sigma^2 d}{4}
   \Big)
   + 
   2\exp\Big(
   -\frac{d}{32}
   \Big)
   + 2
   \exp \Big(
   \frac{- \lambda^3 \sigma^2 d }{128 L^2 } \Big)
   +
   2\exp\Big(-\frac{d}{8}\Big) \\
   \leq &
   3 \exp \Big(
   \frac{- \lambda^3 \sigma^2 d }{128 L^2 }
   \Big) + 
   4\exp\Big(
   -\frac{d}{32}
   \Big).
\end{align*}
This concludes the proof.
\end{proof}
\begin{lemma}\label{lemma:chiq} 
Let $W\sim \chi_d^2$ and $c>0$
    \begin{equation*}
    \E_W  \exp\left( -\frac{c}{W} \right) \leq \exp\left(-\frac{c}{2d}\right) + \exp\left(-\frac{d}{8}\right) ,
    \end{equation*}
\end{lemma}
\begin{proof}
    Since $\exp(-c/W)$ is positive and smaller than 1, we write the expectation in terms of survival function as
    \begin{equation*}
        \begin{aligned}
            \E_W  \exp\left( -\frac{c}{W} \right) = & \int_{0}^1 \Pr\Big( W > -\frac{c}{\log u} \Big) du \\
            =& \int_{0}^1 \Pr\Big( \frac{1}{d}\sum_{i =1}^n Z_i^2 - 1 > -\frac{c}{d \log u} - 1  \Big) du, 
        \end{aligned}
    \end{equation*}
    where $Z_i, \, i = 1,\dots,n,$ are i.i.d standard normal random variables. Note that $ -c/(d \log u) - 1 > 1 $ if $u > u^* = \exp(-c/(2d))$. Therefore, in view of Eq. (2.18) of \citet{wainwright2019high}, it follows  
    \begin{equation*}
         \E_W  \exp\left( -\frac{c}{W} \right) \leq \exp\Big(- \frac{c}{2d}\Big) + \int_{u^*}^1 \Pr\Big( \frac{1}{d}\sum_{i =1}^n Z_i^2 - 1 > -\frac{c}{d \log u} - 1  \Big) \leq \exp\Big(- \frac{c}{2d}\Big) + \exp\Big( -\frac{d}{8} \Big),
    \end{equation*}
    which concludes the proof of the lemma.
\end{proof}
\begin{lemma} \label{lemma:marginal:convex:smooth} 
Let $x \sim \pi $ with $\pi$ satisfying Assumption \ref{cond:1}. Let $\bV \in \V_{m}(\R^d)$ and consider the linear transformation $g_{\bV} \, : \, \R^d \to \R^m \times \R^{d-m}$ such that $g_{\bV}(x) = (s^{\perp}, s)$ where $s^{\perp} = (\bV^{\perp})^{\top} x$, with $\bV^{\perp}$ defined as in Lemma \ref{lemma:prop:smala}, and $s = \bV^{\top} x$. Then, the marginal distribution of $s^{\perp}$, $\pi_{\bV}$, satisfies Assumption \ref{cond:1} for the same constants $\lambda$ and $L$.
\end{lemma}
\begin{proof}
    $L$-smoothness and $\lambda$-convexity are preserved under linear orthonormal transformations. Therefore, it is sufficient to prove the lemma for a generic partition $(x_1, x_2)$ of $x$, where $x_1 \in \R^{d-m}$ and $x_2 \in \R^m$. Let us denote with $\pi_1$ the marginal distribution of $x_1$ and with $ U_1(x_1) = - \log \big( \int_{\sX_2} \pi(x_1,d x_2)\big) $ the potential of $\pi_1$. Under Assumption \ref{cond:1}, $\lambda$-convexity of $U_1(x)$ is proven in \citet[][Theorem 3.8]{saumard2014log}.
    To prove $L$-smoothness, let 
    $$\pi(dx_2 \mid x_1)
    = 
    \pi(x_1,dx_2) /  \int_{\sX_2} \pi(x_1,dx_2).
    $$
    By applying the Leibniz integral rule and simple algebraic manipulations it follows that the first and second derivatives of $U_1(x)$ are equal to 
    \begin{equation*}
        \frac{\partial}{\partial x_1} U_1(x_1)
        =
        \int_{\sX_2} \left(
        \frac{\partial}{\partial x_1} 
        U(x_1,x_2)
        \right) 
        \pi(dx_2 \mid x_1),
    \end{equation*}
    and 
    \begin{equation} \label{eq:marginal:2nd}
     \frac{\partial^2}{\partial x_1 \partial x_1^{\top}} U_1(x_1)
     =
     \int_{\sX_2} \left(
     \frac{\partial^2}{\partial x_1  \partial x_1^{\top} } 
     U(x_1,x_2)
     \right) 
    \pi(dx_2 \mid x_1)
        -
        \mathrm{Cov}\left(
        \frac{\partial}{\partial x_1} 
        U(x_1,x_2)
        \right),
    \end{equation}
    where, for a given $x_1 \in \R^m$, the last term in the right-hand-side of \eqref{eq:marginal:2nd} denotes the covariance matrix of $(\partial/\partial x_1) U(x_1,x_2)$ under the conditional law $\pi(dx_2 \mid x_1)$. From Assumption \ref{cond:1} and Lemma C.3 in the Supplementary Material of \citet{durante2024skewed} it follows that $ ( \partial^2/(\partial x_1  \partial x_1^{\top} )) 
     U(x_1,x_2) \preccurlyeq L \I_{m}$ while 
     $\mathrm{Cov}
     \big(
        (\partial/\partial x_1) 
        U(x_1,x_2)
    \big)$
    is, by construction, positive semidefinite. This implies
    \begin{equation*} 
     \frac{\partial^2}{\partial x_1 \partial x_1^{\top}} U_1(x_1)
     \preccurlyeq
     L \I_{m},
    \end{equation*}
    which concludes the proof of the lemma.
\end{proof}

\begin{lemma}{(\textbf{Distribution on the sphere induced by the l2 norm})} \label{lemma:l2:sphere}
Let $p(x)$ be a probability density function with support on $\mathbb{R}^d$ and, for every $x \in \R^d$, define the map 
\begin{equation*}
 (x_1, \dots x_d) \to ( y_1, \dots, y_{d-1}, r),   
\end{equation*}
where $y_i = x_i/\|x\|_2,\, i = 1,\dots,d-1,$ and $r = \sign(x_d) \|x\|_2$ with $\|x\|_2 = ( \sum_{i = 1}^d |x_i|^2)^{1/2}$.

Then, the probability density function of the vector $(y,r)$ takes the form
\begin{equation*}
  p(y,r) = p(x(y,r)) |r|^{d-1}(1-\|y\|^2_2)^{-1/2},
\end{equation*}
where 
\begin{equation} \label{expr:inv:x}
 x(y,r) = (|r| y_1, |r| y_2, \dots, r (1-\|y\|^2_2)^{1/2}).   
\end{equation}
\end{lemma} 
\begin{proof}
   The result follows directly from the change of variables formula for probability density functions.
\end{proof}

\begin{lemma}\label{lemma:norm:ucon}
Let $U$ be $\lambda$-convex with minimum at $x_{*} \in \R^d$. Then $\|X-x^*\|^2 \overset{s}{\preccurlyeq}  \lambda^{-1}\chi ^2_d$ where $\overset{s}{\preccurlyeq}$ denotes the stochastic ordering.
\end{lemma}
\begin{proof}
    Without loss of generality assume $x^* =0$. Let $X \sim \pi$, with density $\pi(x) \propto \exp(-U(x))$. Let $Z \sim N_d(0, \lambda^{-1} \I_d)$. To prove the statement consider the random variables $R_X = \sign(X_{d}) \|X\|_2$, $R_Z = \sign(Z_{d}) \|Z\|_2$ and note that $ |R_Z| \sim \sqrt{ \lambda^{-1}\chi ^2_d}$. Let  $Y_{X,i} = X_{i}/\|X\|_2,$ $i = 1, \dots,d-1$. In view of Lemma \ref{lemma:l2:sphere}, the conditional probability density function of $R_X$ given $Y_X = (Y_{X,1}, \dots,Y_{X,d} )$ is $p_{R_X|Y_X=y}(r) \propto \exp( -U( x(y,r)) + (d-1)\log(|r|))$ with $x(y,r)$ defined in \eqref{expr:inv:x}. Similarly, the marginal probability density function of $R_Z$ satisfies $p_{R_Z}(r) \propto \exp(-\lambda r^2/2 + (d-1)\log( |r|)).$ The statement of the lemma is proven by showing that $R_X | \{ R_X \geq 0\} \overset{s}{\preccurlyeq} R_Z | \{ R_Z \geq 0\}$ and $R_X |\{R_X < 0 \} \overset{s}{\succcurlyeq} R_Z | \{ R_Z < 0\}$. 
    To do that, we first study the behavior of the ratio
    \begin{equation*}
            \frac{p_{R_X}(r) \mathbbm{1}_{r>0}}{p_{R_Z}(r)} =\int  \frac{p_{R_X|Y_X=y}(r) }{p_{R_Z}(r)} p_{Y_X}(dy)  = \int \exp \left \lbrace \frac{\lambda r^2}{2} -U(x(r,y) )\right \rbrace  p_{Y_X}(dy). 
    \end{equation*}
    The derivative of the quantity reported in the last display can be rearranged by applying the Leibniz integral rule as follows 
    \begin{equation*}
        \begin{aligned}
            \frac{\partial }{ \partial r}
            \int
            \exp
            \left\lbrace
            \frac{\lambda r^2}{2} -U(x(r,y))
            \right\rbrace 
            p_{Y_X}(dy)
            =
            \int
            \exp
            \left\lbrace
            \frac{\lambda r^2}{2} -U(x(r,y))
            \right\rbrace
            \Big(
            \lambda r 
            - 
            \frac{\partial }{ \partial r}
            U(x(r,y)
            \Big)
            p_{Y_X}(dy).
\end{aligned}
\end{equation*}
Moreover, let $\tilde y = (y_1,\dots,y_{d-1}, \sqrt{1-\|y\|_2^2})$ and note that 
\begin{equation*}
     \frac{\partial }{ \partial r} U(x(r,y)) \, = \,  \Big \lbrace  \frac{\partial }{ \partial x} U(x(r,y)) \Big \rbrace^{\top} \tilde y, 
\end{equation*}
since $U(x)$ is $\lambda$-convex,
 \begin{equation} 
 \label{eq:2nd:dev:ratio}
     \frac{\partial^2 }{ \partial^2 r} U(x(r,y)) \,=\, \tilde y^{\top}  \Big \lbrace  \frac{\partial^2 }{ \partial x \partial x^{\top}} U(x(r,y)) \Big \rbrace \tilde y \geq \lambda.  
\end{equation}
 As result, using \eqref{eq:2nd:dev:ratio} and for $r \geq 0$,
\begin{equation*}
  \lambda r -  \frac{\partial }{ \partial r} U(x(r,y))
  =
  \, \lambda r - 
  \int_{0}^r
  \tilde y^{\top}
  \Big \lbrace
  \frac{\partial^2 }{ \partial x \partial x^{\top}}
  U(x(r^*,y))
  \Big \rbrace 
  \tilde y
  dr^* 
  \leq
  \, \lambda r - \lambda \int_{0}^r  dr^* \,
  =
  \, 0.
\end{equation*}
This implies that the ratio $p_{R_X}(r)/p_{R_Z}(r)$ is decreasing in $r$ (i.e., $p_{R_Z}(r)/p_{R_X}(r)$ is increasing). Thus, in view of Lemma 5 of \citet{ascolani2023clustering}, $R_X | \{ R_X \geq 0\} \overset{s}{\preccurlyeq} \{R_Z | R_Z \geq 0\}$. Following the same steps, it is possible to prove $R_X |\{ R_X < 0\} \overset{s}{\succcurlyeq} R_Z | \{ R_Z < 0\}$. The combination of these two results proves the statement of the lemma. 
\end{proof}

\begin{lemma} \label{lemma:norm:cond2}
Let $a,b\in\R^d$ and $\bV \sim \nu$ with $ \nu $ satisfying Assumption \ref{cond:2}. Then
\begin{equation*}
    \E \big( 
    \|a + \bV \bV^\top b \|^2
    \big)
    =
    \frac{d-m}{d} \|a\|^2 +
    \frac{m}{d} \| a + b \|^2.
\end{equation*}

\end{lemma}
\begin{proof}
Recall that $\I_d = \bV \bV^\top + \bV^{\perp} (\bV^{\perp} )^\top $ for any $\bV \in \V_m(\R^d)$ Stiefeld. This, and the fact that $\E ( \bV \bV^\top ) = (m/d) \I_d$, imply $\E( \bV^{\perp} (\bV^{\perp} )^\top ) = \{(d-m)/d\} \I_d$. As a consequence, using the identity $a =  \bV \bV^\top a + \bV^{\perp} (\bV^{\perp} )^\top a$, it follows from simple algebraic manipulations that 
\begin{align*}
    \E \big( 
    \|a + \bV \bV^\top b \|^2
    \big)
    = &
    a^\top \left \lbrace
    \E ( \bV \bV^\top ) + 
    \E( \bV^{\perp} (\bV^{\perp} )^\top )
    \right \rbrace 
    a +
    2 a^{\top} \E ( \bV \bV^\top ) b +
     b^{\top} \E ( \bV \bV^\top ) b \\
     = & 
     \frac{(d-m)}{d} \| a \|^2 +
     \frac{m}{d} \| a + b\|^2.
\end{align*}
This concludes the proof of the lemma.
\end{proof}
\end{document}